\documentclass{article}
\usepackage[a4paper]{geometry}
\usepackage{hyperref}
\usepackage{color}
\usepackage{multicol}
\usepackage{subcaption}
\usepackage{wrapfig}
\usepackage{amsmath}
\usepackage{amsthm}
\usepackage{tikz}
\usepackage{amsmath}
\usepackage{amsfonts}
\usepackage{epstopdf}
\usetikzlibrary{shapes,arrows,automata}
\usepackage{tikz-cd}
\usepackage{multirow}
\usepackage{appendix}

\newtheorem{theorem}{Theorem}[section]

\title{A Dynamical Framework for Modeling Fear of Infection and Frustration with Social Distancing in COVID-19 Spread}
\author{Matthew D. Johnston and Bruce Pell\\ \\
Lawrence Technological University\\
21000 W 10 Mile Rd.\\
Southfield, MI 48075\\
{\tt mjohnsto1@ltu.edu}\\
{\tt bpell@ltu.edu}}
\date{\today}

\usepackage{graphicx}

\begin{document}

\maketitle

\begin{abstract}

In this paper, we introduce a novel modeling framework for incorporating fear of infection and frustration with social distancing into disease dynamics. We show that the resulting SEIR behavior-perception model has three principal modes of qualitative behavior---\emph{no outbreak}, \emph{controlled outbreak}, and \emph{uncontrolled outbreak}. We also demonstrate that the model can produce transient and sustained waves of infection consistent with secondary outbreaks. We fit the model to cumulative COVID-19 case and mortality data from several regions. Our analysis suggests that regions which experience a significant decline after the first wave of infection, such as Canada and Israel, are more likely to contain secondary waves of infection, whereas regions which only achieve moderate success in mitigating the disease's spread initially, such as the United States, are likely to experience substantial secondary waves or uncontrolled outbreaks.
\end{abstract}

\section{Background}

Since being first detected in Wuhan, China, in December 2019, severe acute respiratory syndrome coronavirus 2 (SARS-CoV-2), and the resulting disease COVID-19, has spread rapidly around the globe. Without a vaccine or generally effective treatment, the unprecedented international mitigation effort has instead focused on reducing transmission through travel restrictions, mandatory quarantines, work-from-home initiatives, school closures, and social distancing practices. These measures have in turn had detrimental effects on economic productivity, job stability, and overall quality of life, which in some regions has limited public willingness to abide by social distancing guidelines, even as cases have continued to rise. To date, COVID-19 has afflicted at least 20 million people worldwide and resulted in over 700,000 deaths \cite{worldometers}.

Mathematical modeling has played a significant role in understanding the primary transmission pathways and epidemiological parameters of the COVID-19 pandemic. Studies have estimated the basic reproductive number, a key measure of the transmissibility of a disease, of SARS-CoV-2  \cite{Kucharski2020,zhang2020,Zhuang2020,Lv2020} and projected the disease's spread under a wide variety of public policy intervention scenarios, including variances in social distancing policies, travel restrictions, and face mask utilization \cite{Barbarossa2020,bouchnita2020,chang2020,djidjou2020,engbert2020,Eikenberry2020}. Forecasting the extent of the spread of COVID-19, however, has been complicated by many factors, including evidence of asymptomatic spread \cite{Park2020,Bai2020}, issues with parameter identifiability \cite{Roda2020,Massonis2020}, and the uncertain mechanisms by which social behaviors have altered the spread of the disease to date \cite{Holmdahl2020,Petropoulos2020,Wynants2020}.


Several methods have been proposed in the research literature for incorporating and evaluating the role of social perception and behavior changes in dynamical models of emerging infectious disease (see the excellent review paper \cite{Wang2015}). The papers \cite{Epstein2008,Perra2011,Valle2013} incorporate social behavior changes into an SIR compartmental model by dividing the susceptible and infectious classes into individuals who initiate behavior changes to reduce transmission and those who do not, and allowing behavior change to transmit through social contact like a contagion. The papers \cite{Cui2008,Fenichel2011,Fenichel2013,Poletti2013,Sun2011} instead incorporate social distancing behavior by allowing perception to evolve as an independent state-dependent variable and having it feedback directly into the transmission rate. Other studies have focused on regional movement patterns, effectiveness of social distancing, and the summer release of school children followed by their return to classes in the fall~\cite{HerreraValdez2011}. The more recent study \cite{Bauch2020} specific to the study of COVID-19 incorporates public support for social distancing as a function of both infection level and economic losses.

In this paper, we present a novel modeling framework for incorporating social perception and behavior \eqref{social} into a compartmental SEIR model \eqref{seir}. Our model incorporates the effects of \emph{fear of infection} ($P_I$) and \emph{frustration with social distancing} ($P_{\omega}$) on \emph{social distancing}  ($\omega$), and the effects of social distancing on disease dynamics by modifying the transmission rate ($\beta$).  Analysis of the corresponding system of differential equations \eqref{feedback-de} suggests that fear of infection can be an effective mitigator of disease spread but that a high level of frustration with social distancing can overwhelm these efforts and result in an uncontrolled outbreak (see Figure \ref{figure1}). Our analysis also suggests that delays in social feedback can lead to transient and sustained waves of infection, even in populations where the disease's spread is controlled (see Figure \ref{figure2}).

We fit the model to cumulative COVID-19 case and mortality data across several regions: Canada, the United States, Israel, Michigan, California and Italy. (See Figures \ref{figure4} and \ref{figure5}.) Our analysis suggests that, although the capacity for secondary waves of infection after controlling the initial outbreak is widespread, the magnitude of the reduction in infection after the initial outbreak is a strong indicator of a society's ability to be able to mitigate the secondary waves. Regions which have significant reductions in infection after the initial outbreak, such as Canada and Israel, are predicted to have modest and controllable secondary waves. Countries which had only modest reductions in infection levels after the initial outbreak, such as the United States, are predicted to have large secondary waves which threaten to become uncontrolled outbreaks. 

This paper is organized as follows. In Section \ref{sec:model} we develop the SEIR behavior-perception system \eqref{feedback-de}, which incorporates social distancing, fear of infection, and frustration with social distancing as time-dependent variables capable of influencing the dynamics of disease spread. In Section \ref{sec:analysis}, we analyze the SEIR behavior-perception system \eqref{feedback-de} and demonstrate the admissible behaviors in a variety of parameter regions, leading to controlled outbreaks, uncontrolled outbreaks, and sustained waves of infection. In Section \ref{sec:data}, we fit the model to cumulative COVID-19 case and mortality data from several regions to estimate key epidemiological and social perception and behavior parameters of the COVID-19 pandemic. In Section \ref{sec:conclusions}, we summarize our results and avenues for future work. In Appendix \ref{app:stability}, we present the mathematical details of the stability results found in Section \ref{sec:analysis}.

\section{Model}
\label{sec:model}

In this section, we introduce an SEIR behavior-perception feedback model of disease spread. The model consists of two parts: (a) an SEIR model for tracking the evolution of disease dynamics \eqref{seir}; and (b) a behavior-perception loop involving social distancing, fear of infection, and frustration with social distancing \eqref{social}.

\subsection{SEIR Model}

We consider the classic \emph{SEIR model} where the population is divided in four compartments: $S$ - \emph{Susceptible} to infection;  $E$ - \emph{Exposed} to infection but not yet symptomatic; $I$ - Actively \emph{infectious}; $R$ - \emph{Removed} from the infection \cite{Kermack1927}. This gives the following model:\\[-0.1in]
\begin{equation}
    \label{seir}
\begin{tikzcd}
\mbox{\fbox{\begin{tabular}{c}Susceptible \\($S$)\end{tabular}}} \arrow[r,"\beta"{name=beta}] & \mbox{\fbox{\begin{tabular}{c}Exposed\\($E$)\end{tabular}}} \arrow[r,"\lambda"] & \mbox{\fbox{\begin{tabular}{c}Infectious\\($I$)\end{tabular}}} \arrow[bend right=60,dashed,from=1-3,to=beta] \arrow[r,"\gamma"] & \mbox{\fbox{\begin{tabular}{c}Removed\\($R$)\end{tabular}}}
\end{tikzcd}
\end{equation}
For simplicity, we assume that only actively infectious individuals can transmit the infection to susceptible individuals, although we note that asymptomatic spread is suspected in COVID-19 \cite{Park2020,Bai2020}. We assume that removed individuals cannot become susceptible to the illness again. It is currently unclear whether COVID-19 may be contracted multiple times but reinfection is not believed to be a significant factor in disease spread \cite{Roy2020}. 

The dynamics of the system \eqref{seir} can be modeling by the following \emph{SEIR system}:
\begin{equation}
    \label{seir-de}
\left\{ \; \; \;
\begin{split}
    \frac{dS}{dt} &  = - \frac{\beta}{N} S I \\
    \frac{dE}{dt} & = \frac{\beta}{N} S I - \lambda E \\
    \frac{dI}{dt} & = \lambda E - \gamma I \\
    \frac{dR}{dt} & = \gamma I 
\end{split}
\right.
\end{equation}
where the parameters are as in Table \ref{table4}. Notice that the recovery/infectious period $\gamma^{-1}$ is the expected number of days until an individual is no longer infectious, regardless of whether that is a result of recovery, death, quarantine, or another means.

Although the SEIR system \eqref{seir-de} cannot be solved explicitly, the dynamics are well-understood. Trajectories with non-negative initial conditions stay non-negative and satisfy the population size conservation equation $N = S(t) + E(t) + I(t) + R(t)$. Trajectories asymptotically approach the steady state $(\bar{S},\bar{E},\bar{I},\bar{R}) = (N-R^*,0,0,R^*)$ where $R^*$ is the solution of $N - R^* = e^{-\frac{\beta}{N\gamma} R^*}$. The value of $R^*$ is the extent of the disease since it corresponds to the number of people who contracted the disease during its course.

The critical parameter for determining whether a disease will spread or die in a population is the \emph{basic reproductive number}, $\mathcal{R}_0$, which quantifies the expected number of secondary infections produced by one active infection entering a fully susceptible population. When $\mathcal{R}_0 > 1$ we expect an outbreak and when $\mathcal{R}_0 < 1$ we expect no outbreak. For the SEIR model \eqref{seir} and corresponding system \eqref{seir-de}, the basic reproductive number is $\mathcal{R}_0 = \frac{\beta}{\gamma}$, which can be calculated using the next-generation method \cite{Diekmann1990,VANDENDRIESSCHE2002,Heffernan2005}. It follows that there will be an outbreak in \eqref{seir-de} if $\beta > \gamma$ while the disease will dissipate without an outbreak if $\beta < \gamma$.

\subsection{SEIR Model With Behavior-Perception Feedback}

The SEIR model \eqref{seir} by itself does not account for the possibility that individuals may alter their behavior, and therefore the disease's trajectory, in response to the disease's spread. Since there is no vaccine or generally effective treatment for the novel coronavirus SARS-CoV-2,  mitigation efforts have necessarily focused on modifications to social behavior such as social distancing, hand washing, and face mask utilization. 



We extend the SEIR model \eqref{seir-de} to incorporate the effects of social perception and behavior change over time. We introduce the following time-dependent variables: \emph{social distancing behavior} ($0 \leq \omega \leq 1$); \emph{perceived fear of infection} ($0 \leq P_I \leq 1$); and \emph{perceived frustration with social distancing} ($0 \leq P_\omega \leq 1$) and assume the following network of dependencies:

\begin{equation}
    \label{social}
\begin{tikzcd}
\mbox{\fbox{\begin{tabular}{c}Spread of\\disease\\$(\lambda E)$\end{tabular}}} \arrow[d,dashed,"+\; \; "'] & \mbox{\fbox{\begin{tabular}{c}Social\\distancing\\($\omega$)\end{tabular}}} \arrow[d,dashed,xshift=0.5ex,"\; \;+"] \arrow[l,dashed,"-"'] \\ \mbox{\fbox{\begin{tabular}{c}Fear of\\infection \\($P_I$)\end{tabular}}}  \arrow[ru,dashed,"+"] & \mbox{\fbox{\begin{tabular}{c}Frustration with\\social distancing\\($P_\omega$)\end{tabular}}} \arrow[u,dashed,xshift=-0.5ex,"- \; \;"]
\end{tikzcd}
\end{equation}
Each arrow indicates how the first quantity influences the second, with a positive label (+) indicating a positive influence and (-) indicating a negative influence.

\begin{table}[t!]
    \centering
    \begin{tabular}{|l|c|l|}
    \hline
    Variable & Units & Description \\
    \hline \hline
    $S \geq 0$ & people & Susceptible individuals \\
    $E \geq 0$ & people & Exposed (non-infectious) individuals \\
    $I \geq 0$ & people & Infectious individuals \\
    $R \geq 0$ & people & Removed individuals \\
    $\omega \in [0,1]$ & \% & Measure of social distancing behavior \\
    $P_I \in [0,1]$ & none & Measure of socially perceived fear of infection \\
    $P_{\omega} \in [0,1]$ & none & Measure of socially perceived frustration with social distancing \\
    $t \geq 0$ & days & Time since start of simulation \\
    \hline \hline
    Parameter & Units & Description  \\
    \hline \hline
    $\beta \geq 0$ & days$^{-1}$ & Transmission rate \\
    $\lambda^{-1} \geq 0$ & days & Incubation period \\
    $\gamma^{-1} \geq 0$ & days & Recovery/infectious period \\
    $k_{\omega} \geq 0$ & days$^{-1}$ & Rate of social distancing behavior change \\
    $k_{P_I} \geq 0$ & days$^{-1}$ & Rate of socially perceived fear of infection change \\
    $k_{P_{\omega}} \geq 0$ & days$^{-1}$ & Rate of socially perceived frustration with socially distancing change \\
    $N \geq 0$ & people & Total population size ($N = S + E +I + R$) \\
    $M \geq 0$ & people & Threshold value of new daily infections for fear of infection\\
    $q \geq 1$ & none & Steepness parameter for fear of infection change \\
    $\omega^* \in [0,1]$ & \% & Maximum reduction to social distancing due to frustration \\
    \hline
    \end{tabular}
    \caption{\small Variables and parameters for the SEIR system \eqref{seir-de} and the SEIR behavior-perception system \eqref{feedback-de}. The parameters $k_{\omega}$, $k_{P_I}$, and $k_{P_{\omega}}$ control the rate of social perception and behavior change. The parameter $M > 0$ is the threshold level of new infections governing social perception change while $q \geq 1$ controls the steepness of the switch from low perceived fear of infection to high perceived fear of infection through the Hill function \eqref{g}. The value of $\omega^*$ controls the extent to which frustration with social distancing reduces social distancing behavior.}
    \label{table4}
\end{table}

We now extend the SEIR system \eqref{seir-de} to include social perceptions and behavior. We assume that the rates of changes of $\omega$, $P_I$, and $P_\omega$ are influenced by variables with arrows leading to it in \eqref{social} and that, in the absence of disease, $\omega$, $P_I$, and $P_\omega$ will decay to zero. This gives the following \emph{SEIR behavior-perception system}:
\begin{equation}
    \label{feedback-de}
    \left\{ \; \; \;
    \begin{split}
    \frac{dS}{dt} &  = - \frac{\beta(1-\omega)}{N} S I \\
    \frac{dE}{dt} & = \frac{\beta(1-\omega)}{N} S I - \lambda E \\
    \frac{dI}{dt} & = \lambda E - \gamma I \\
    \frac{dR}{dt} & = \gamma I 
\end{split} \hspace{0.5in}
    \begin{split}
    \frac{d\omega}{dt} & = k_{\omega} \left( f(P_I,P_\omega) - \omega \right) \\
    \frac{dP_I}{dt} & = k_{P_I} \left( g(\lambda E) - P_I \right) \\
    \frac{dP_\omega}{dt} & = k_{P_\omega} \left( h(\omega) - P_\omega \right).
    \end{split}
    \right.
\end{equation}
\noindent 
Notice that the social perception variables $P_I$ and $P_{\omega}$ influence the social behavior variable $\omega$ which in turn influences the effective transmission rate $\beta(1-\omega)$, as outlined in \eqref{social}. Notice also that we use the daily number of new active cases $\lambda E$ as the catalyst for social perception change rather than the number of current active infections $I$, as used in \cite{Bauch2020}. We believe $\lambda E$ correlates better with publicly reported case incidence data than the current infection level $I$.

\begin{figure}[t!]
    \centering
    \includegraphics[width=6in]{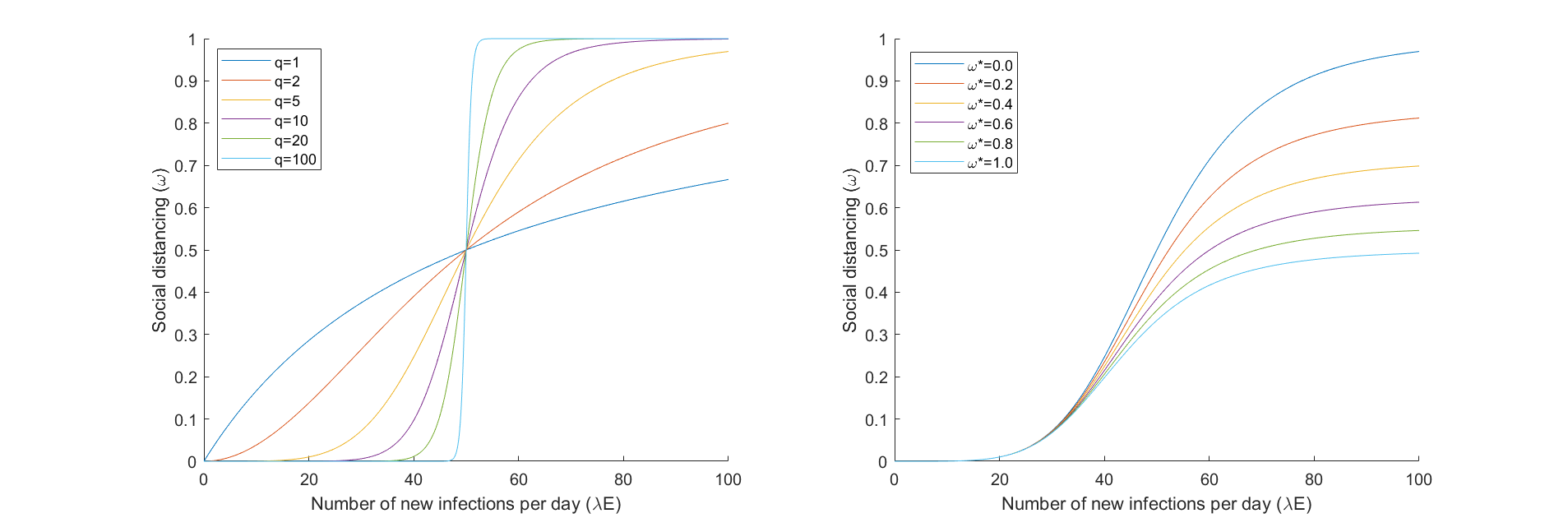}
    \caption{\footnotesize The relationship between new daily infections ($\lambda E$) and social distancing ($\omega$) in the SEIR behavior-perception model \eqref{feedback-de}, assuming the social perception and behavior variables $\omega$, $P_I$, and $P_\omega$ are at quasi-steady state \eqref{omega}. On the left, we take $M = 50$, $\omega^* = 0$, and vary $q$. As $q$ increases, the response of social distancing to new infections becomes sharper. On the right, we take $M = 50$, $q = 5$, and vary $\omega^*$. As $\omega^*$ increases, the effective level of social distancing decreases even when the level of new infections is high ($\lambda E > M$). Note that enforcing the bounds $\omega^* \in [ 0, 1 ]$ guarantees $\omega(t) \in [0, 1 ]$ for all $t \geq 0$ but that the social distancing response \eqref{omega} has a lower bound of $0.5$ even when new infections are high. To obtain a fully saturated quasi-steady state $\omega$ value lower than $0.5$, we can take $\omega^* > 1$ but note that this may yield a negative transient $\omega(t)$ value.}
    \label{figure6}
\end{figure}

We utilize the following functions $f$, $g$, and $h$:
\begin{eqnarray}
    \label{f}
    f(P_I,P_\omega) & = & P_I( 1 - \omega^* P_\omega) \\
    \label{g}
    g(\lambda E) & = &  \frac{(\lambda E)^q}{M^q + (\lambda E)^q}\\
\label{h}
    h(\omega) & = &  \omega. 
\end{eqnarray}
We use a Hill function \cite{Hi} on the level of new infections $\lambda E$ to control the perceived fear of infection \eqref{g} and note that this function become more ``switch-like'' around the threshold value $M$ as $q$ grows (see Figure \ref{figure6}). The parameters $k_{\omega}$, $k_{P_I}$, $k_{P_{\omega}}$, $q$, $M$, and $\omega^*$ are as in Table \ref{table4}.

To further illustrate how social perception and behavior influences disease dynamics, we take $\omega$, $P_I$, and $P_\omega$ at quasi-steady state in \eqref{feedback-de}. This gives the following relationship for the level of new daily infections $\lambda E$ on the social distancing variable $\omega$:
\begin{equation}
    \label{omega}
    \omega = \frac{(\lambda E)^q}{M^q + (1 + \omega^*) (\lambda E)^q}.
\end{equation}
Illustration of the effect of the parameters $q$, $M$, and $\omega^*$ on the relationship \eqref{omega} is contained in Figure \ref{figure6}. 

\section{Model Analysis}
\label{sec:analysis}

In this section, we analyze the SEIR behavior-perception system \eqref{feedback-de} by considering reduced models which simulate the early-stage dynamics of the outbreak. We show that there are three dominants modes of behavior---\emph{no outbreak}, \emph{controlled outbreak}, and \emph{uncontrolled outbreak}. We also show that the system \eqref{feedback-de} has the capacity for transient and sustained waves of infection, and investigate how this depends upon the delays in the social perception and behavior changes.

\subsection{Reduced Models}
\label{sec:reduced}


We reduce the behavior-perception model \eqref{feedback-de} to consider the early-stage dynamics when nearly everyone in the population is susceptible to illness. To accomplish this, we set $S = N$ and remove the equations for $S$ and $R$. To further investigate the mechanisms which contribute to secondary waves of infection, we consider three scenarios on the delays in social perception---\emph{no delays}, \emph{one delay}, and \emph{two delays}. In all cases, we state the model with frustration with social distancing included but can remove this variable by taking $\omega^* = 0$.

\paragraph{Model I: No delays.} We assume that the social perception and behavior variables in \eqref{feedback-de} operate on a significantly faster time scale than the disease dynamic variables. On the slow-time scale of the disease dynamics variables, this corresponds to taking $\omega$, $P_I$, and $P_{\omega}$ in \eqref{feedback-de} at steady state, which yields \eqref{omega}. 
Substituting \eqref{omega} into \eqref{feedback-de} and  removing $S$ and $R$ gives the following \emph{reduced no delay SEIR behavior-perception system}:
\begin{equation}
    \label{reduced-direct}
\left\{ \; \; \begin{split} \frac{dE}{dt} & = \beta\left( 1- \frac{(\lambda E)^q}{M^q + (1 + \omega^*) (\lambda E)^q}\right) I - \lambda E \\ \frac{dI}{dt} & = \lambda E - \gamma I.  
\end{split} \right.
\end{equation}
The behavior of this model is a reasonable approximation of that of \eqref{feedback-de} so long as the outbreak remains small ($S \approx N$) and the social perceptions and behavior variables evolve on a significantly faster time-scale than the disease dynamics variables.

\paragraph{Model II: One delay.} We assume that fear of infection evolves on a faster time scale than the remainder of the variables, which corresponds to taking $P_I$ at steady state in \eqref{feedback-de}. After removing $S$ and $R$, this gives the \emph{reduced one delay SEIR behavior-perception system}:
    \begin{equation}
    \label{seir-onedelay-reduced}
    \left\{ \; \;
    \begin{split}
    \frac{dE}{dt} & = \beta\left(1 - \omega \right) I - \lambda E \\
    \frac{dI}{dt} & = \lambda E - \gamma I\\
    \end{split}  \hspace{0.5in}
    \begin{split}
    \frac{d\omega}{dt} & = k_{\omega} \left( \frac{(\lambda E)^q}{M^q + (\lambda E)^q}\left( 1 - \omega^* P_{\omega}\right) - \omega \right) \\
    \frac{dP_{\omega}}{dt} & = k_{P_{\omega}} \left( \omega - P_{\omega} \right).
    \end{split}
    \right.
\end{equation}
This model approximates \eqref{feedback-de} when the outbreak remains small ($S \approx N$) and social perception of the disease evolves on a significantly faster time-scale than social behavior change or frustration due to social distancing. This is appropriate in societies where information is readily available due to media and other vectors of mass communication, but will or capability to change social behavior is limited.

\paragraph{Model III: Two delays.} We remove $S$ and $R$ but do not assume any social perception or behavior variables operate on a significantly faster time scale than the disease spread variables. This gives the \emph{reduced two delay SEIR behavior-perception system}:
\begin{equation}
    \label{reduced-de}
    \left\{ \; \; \;
    \begin{split}
    \frac{dE}{dt} & = \beta(1-\omega)I - \lambda E \\
    \frac{dI}{dt} & = \lambda E - \gamma I \\
    \frac{d\omega}{dt} & = k_{\omega} \left( P_I(1 - \omega^* P_\omega) - \omega \right)
\end{split} \hspace{0.5in}
    \begin{split}
    \frac{dP_I}{dt} & = k_{P_I} \left( \frac{(\lambda E)^q}{M^q + (\lambda E)^q} - P_I \right) \\
    \frac{dP_\omega}{dt} & = k_{P_\omega} \left( \omega - P_\omega \right)
    \end{split}
    \right.
\end{equation}
This model is the closest approximation of \eqref{feedback-de} since it only assumes that changes in $S$ occur on a slower time scale than the remainder of the variables. The system \eqref{reduced-de} is appropriate for studying the early stages of an outbreak, diseases which are not easily transmitted, or diseases which are significant controlled through social intervention and never enter a mode of true outbreak.



\subsection{Controlled and Uncontrolled Outbreaks}
\label{sec:controlled}

\begin{figure}[t!]
\begin{subfigure}{.32\textwidth}
\centering
\includegraphics[width=2in]{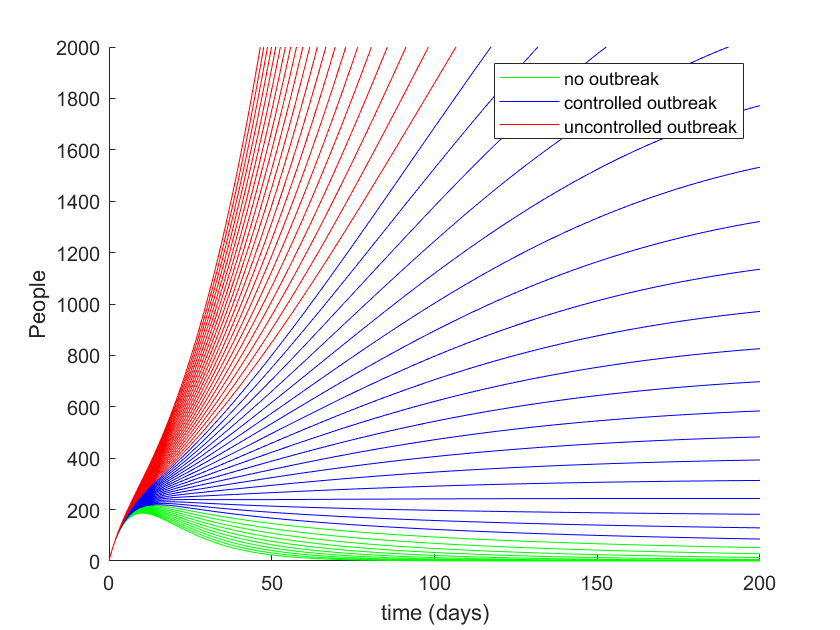}
\includegraphics[width=2in]{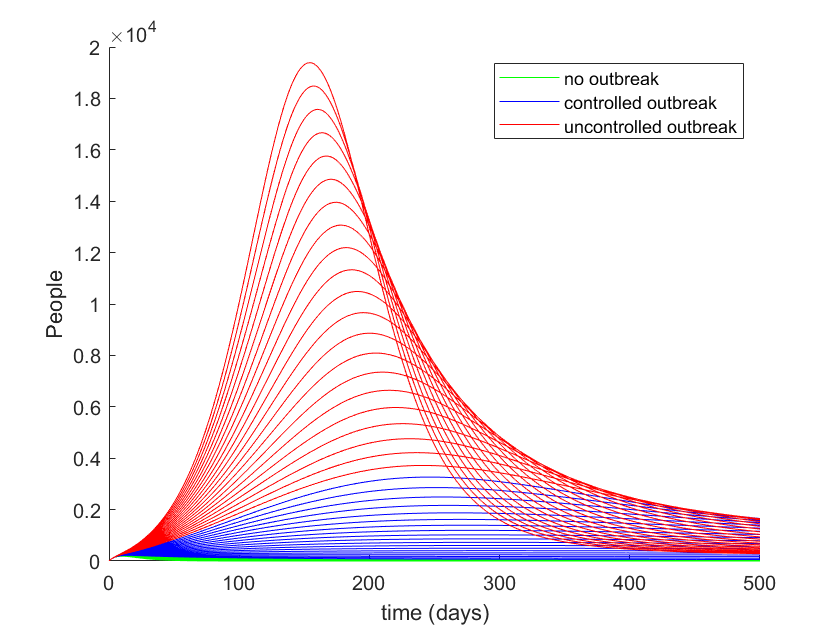}
\caption{$q = 1$}
\end{subfigure}
\begin{subfigure}{.32\textwidth}
\centering
\includegraphics[width=2in]{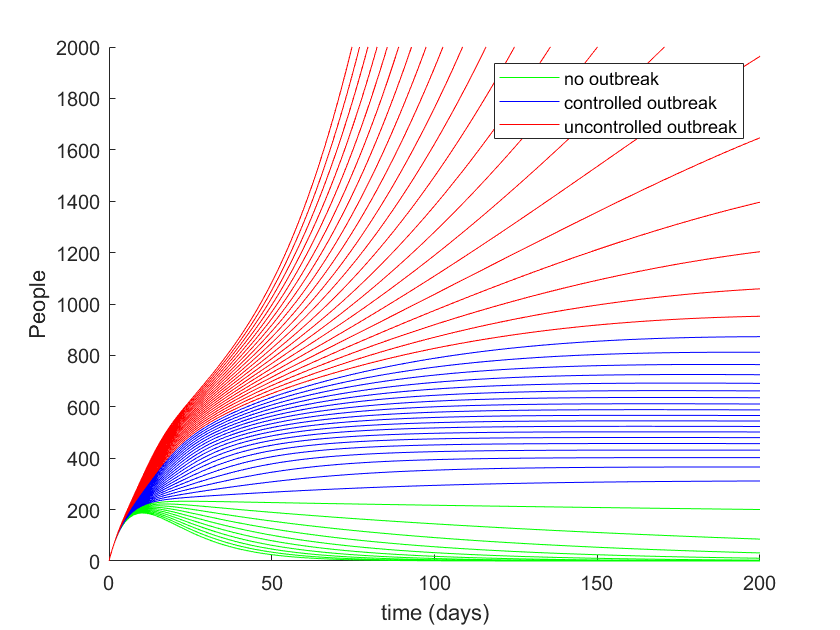}
\includegraphics[width=2in]{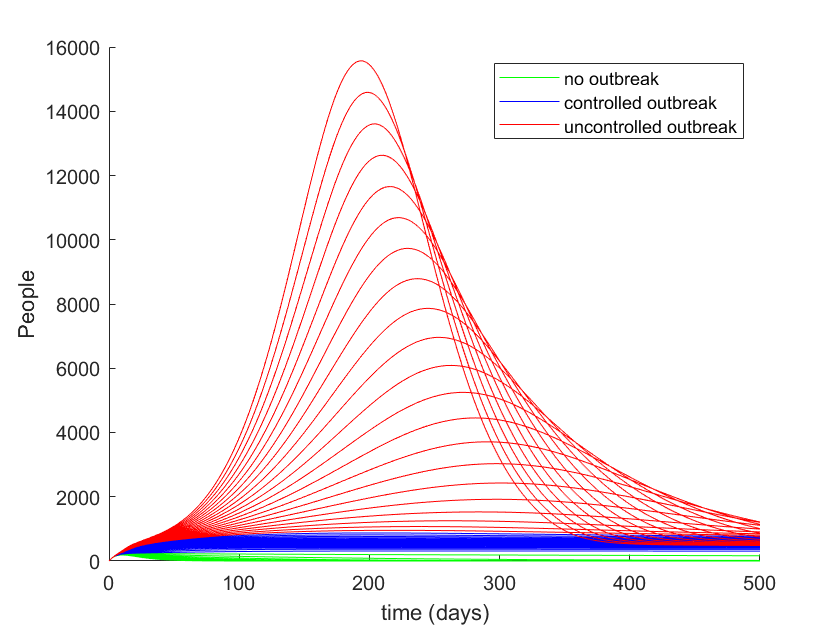}
\caption{$q = 5$}
\end{subfigure}
\begin{subfigure}{.32\textwidth}
\centering
\includegraphics[width=2in]{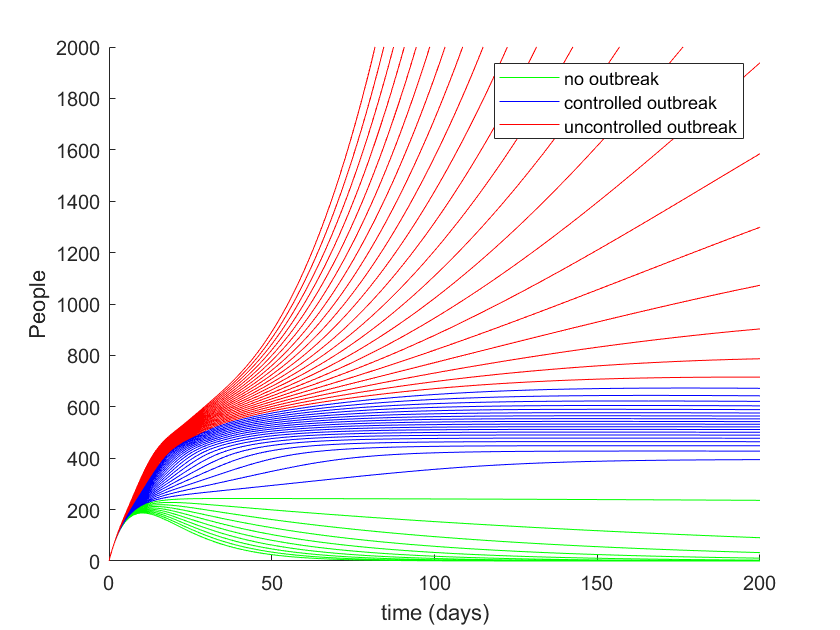}
\includegraphics[width=2in]{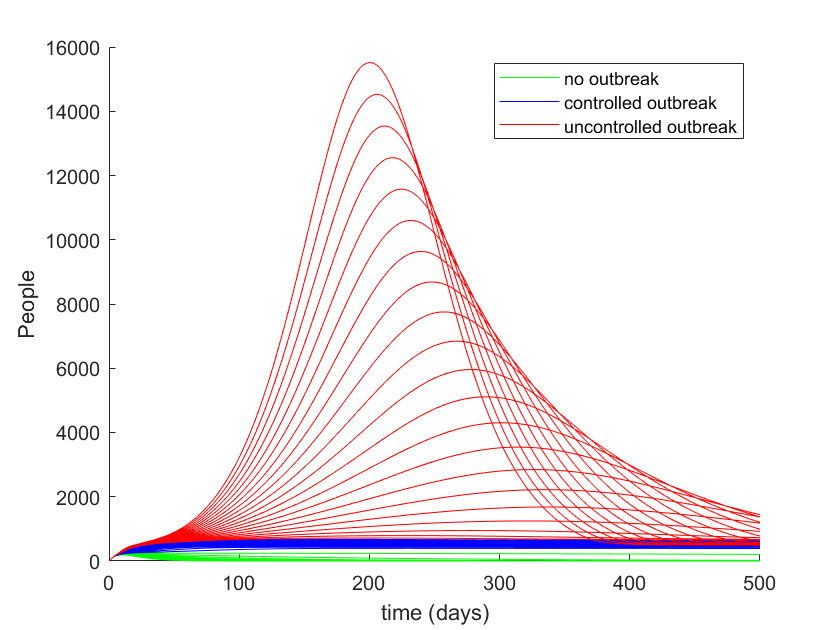}
\caption{$q = 10$}
\end{subfigure}
\caption{\footnotesize Simulations of the active infection level of the SEIR behavior-perception system \eqref{feedback-de} with parameter values $\lambda = 0.1$, $\gamma = 0.1$,  $M = 50$, $k_{\omega} = 0.5$, $k_{P_I} = 0.5$, $k_{P_{\omega}} = 0.05$, $\omega^* = 0.5$, and $N = 1000000$, and initial conditions $S(0) = 999500$, $E(0) = 500$, $I(0) = 0$, $R(0) = 0$, $\omega(0) = 0$, $P_I(0) = 0$, and $P_{\omega}(0) = 0$. Simulations were conducted for $\beta$ values from $0.01$ to $0.5$ taken in increments of $0.01$ taking three different values of $q$: (a) $ q= 1$; (b) $q = 5$; and (c) $q = 10$. The results are divided into three classes according to Table \ref{table3}: (i) no outbreak, $0 < \beta < 0.1$, (green); (ii) controlled outbreak, $0.1 \leq \beta < 0.3$ (blue); and (iii) uncontrolled outbreak, $0.3 \leq \beta \leq 0.5$ (red). Notice that the sharpness parameter $q$ does not affect whether there is controlled or uncontrolled outbreak but does control the distinction between trajectories in the respective regions.}
\label{figure1}
\end{figure}

The primary epidemiological parameter for determining whether an outbreak will occur is the \emph{basic reproductive number} $\mathcal{R}_0$. This value corresponds to the expected number of secondary infections resulting from a single primary infection in a fully susceptible population (i.e. early in an outbreak when $S \approx N$). Consequently, when $\mathcal{R}_0 < 1$ it is predicted that the disease will die off before an outbreak occurs, while if $\mathcal{R}_0 > 1$ it is predicted that the disease will spread. 

To determine the basic reproductive number, we consider the \emph{disease-free steady state} \\ $(\bar{E},\bar{I},\bar{\omega},\bar{P_I},\bar{P_{\omega}})_{df} = (0,0,0,0,0)$ of the reduced models \eqref{reduced-direct}, \eqref{seir-onedelay-reduced}, and \eqref{reduced-de}, and apply the next-generation method \cite{Diekmann1990,VANDENDRIESSCHE2002,Heffernan2005}. We will use this single steady state for all three reduced models by restricting to the relevant model variables as required. For all three models, it can be computed that the dominant eigenvalue of the next-generation matrix at the disease-free state is $\frac{\beta}{\gamma}$ so that $\mathcal{R}_0 = \frac{\beta}{\gamma}$. It follows that the disease will die off if $\beta < \gamma$ and we will have an outbreak if $\beta > \gamma$.

\begin{table}[t!]
\footnotesize
\centering
\begin{tabular}{c|c|c|c|c|c}
    \hline \hline
    Model & Behavior & \begin{tabular}{c} Parameter \\ range \end{tabular} & \begin{tabular}{c} Disease-free \\ steady state \end{tabular} & \begin{tabular}{c} Endemic \\ steady state \end{tabular} & Oscillations\\
    \hline \hline
    \multirow{3}{*}{\begin{tabular}{c}no delays\\\eqref{reduced-direct}\end{tabular}} & No outbreak & $\beta < \gamma$ & stable & DNE & \multirow{3}{*}{none}\\
    & Controlled outbreak & $0 < \beta - \gamma < \frac{\gamma}{\omega^*}$ & unstable & stable & \\
    & Uncontrolled outbreak & $\beta - \gamma > \frac{\gamma}{\omega^*}$ & unstable & DNE & \\
    \hline \hline
    \multirow{3}{*}{\begin{tabular}{c}one delay \\ \eqref{seir-onedelay-reduced} \end{tabular}} & No outbreak & $\beta < \gamma$ & stable & DNE & \multirow{3}{*}{transient}\\
    & Controlled outbreak & $0 < \beta - \gamma < \frac{\gamma}{\omega^*}$ & unstable & stable & \\
    & Uncontrolled outbreak & $\beta - \gamma > \frac{\gamma}{\omega^*}$ & unstable & DNE & \\
    \hline \hline
    \multirow{3}{*}{\begin{tabular}{c}two delays\\\eqref{reduced-de}\end{tabular}} & No outbreak & $\beta < \gamma$ & stable & DNE & \multirow{3}{*}{sustained}\\
    & Controlled outbreak & $0 < \beta - \gamma < \frac{\gamma}{\omega^*}$ & unstable & varies$^*$ & \\
    & Uncontrolled outbreak & $\beta - \gamma > \frac{\gamma}{\omega^*}$ & unstable & DNE & \\
    \hline \hline
\end{tabular}
\caption{\footnotesize Summary of the analysis of the disease-free and endemic steady states of the reduced SEIR behavior-perception systems \eqref{reduced-direct}, \eqref{seir-onedelay-reduced}, and \eqref{reduced-de}. The analysis predicts three distinct modes of epidemic behavior: (a) \emph{no outbreak} if $\beta < \gamma$; (b) \emph{controlled outbreak} if $0 < \beta - \gamma < \frac{\gamma}{\omega^*}$; and (c) \emph{uncontrolled outbreak} if $\beta - \gamma > \frac{\gamma}{\omega^*}$.  The distinction between these three cases on system \eqref{feedback-de} is illustrated in Figure \ref{figure1}. The analysis also suggests that delays in social perceptions feeding back into social behavior are required in order to have secondary waves of infection. In particular, for the two delay model, the endemic steady state may lose stability even when the outbreak is controlled, which yields limit cycles. These sustained waves of infection are demonstrated in Figure \ref{figure2}. The mathematical analysis and numerically derived boundaries of stability are contained in Appendix \ref{app:stability}.}
    \label{table3}
\end{table}

The reduced models \eqref{reduced-direct}, \eqref{seir-onedelay-reduced}, and \eqref{reduced-de} also have an \emph{endemic steady state}, which is given by
\begin{equation}
    \label{endemic-ss} \small
(\bar{E},\bar{I},\bar{\omega},\bar{P_I},\bar{P_{\omega}})_{end} = \left( \frac{M}{\lambda} \left( \frac{\beta-\gamma}{\gamma-\omega^*(\beta - \gamma)} \right)^{\frac{1}{q}}, \frac{M}{\gamma} \left( \frac{\beta-\gamma}{\gamma - \omega^*(\beta - \gamma)} \right)^{\frac{1}{q}}, \frac{\beta - \gamma}{\beta}, \frac{\beta - \gamma}{\beta- \omega^*(\beta - \gamma)} , \frac{\beta-\gamma}{\beta}\right).
\end{equation}
Notice that the endemic steady state \eqref{endemic-ss} is only physically meaningful for any of the three reduced models when $0 < \beta - \gamma < \frac{\gamma}{\omega^*}$. Again, we can consider \eqref{endemic-ss} to be the steady state for all three reduced models by restricting to the appropriate model variables.

The endemic steady state \eqref{endemic-ss} corresponds to a controlled outbreak where the the fear of infection keeps the spread of the disease from growing uncontrolled (see Figure \ref{figure1}). 
Notice, however, that if $\omega^*$ is sufficiently high, the endemic steady state is not physically relevant and the system exhibits a full outbreak similar to as though no social interventions had taken place. This suggests that frustration with social distancing can undo the gains in mitigating disease spread given by fear of infection and the resultant social distancing.

To further understand the behavior of the reduced models \eqref{reduced-direct}, \eqref{seir-onedelay-reduced}, and \eqref{reduced-de}, and therefore the full SEIR behavior-perception model \eqref{feedback-de}, we conduct linear stability analysis on the disease-free and endemic steady states of the reduced model (see Appendix \ref{app:stability} for mathematical details). A summary of the analysis is contained in Table \ref{table3}. Simulations of the full SEIR behavior-perception model \eqref{feedback-de} are contained in Figure \ref{figure1} for various values of the transmission rate $\beta$ and sharpness parameter $q$.

These analyses suggest three distinct possibilities for the dynamics:
\begin{enumerate}
    \item[(i)] \textbf{No outbreak} when $\beta < \gamma$ (green in Figure \ref{figure1}). In this case, the overall level of infection $E(t) + I(t)$ tends monotonically to zero. This corresponds to a disease which is not virulent enough to maintain presence in a population.
    \item[(ii)] \textbf{Controlled outbreak} when $0 < \beta - \gamma < \frac{\gamma}{\omega^*}$ (blue in Figure \ref{figure1}). In this case, the infection nears the endemic steady state \eqref{endemic-ss}. This corresponds to a population which, through social perception and behavior feedback, is able to tolerate a certain amount of ambient infection without ever fully ridding it from the population.
    \item[(iii)] \textbf{Uncontrolled outbreak} when $\beta - \gamma > \frac{\gamma}{\omega^*}$ (red in Figure \ref{figure1}). In this case, the infection runs through the population relatively unimpeded and forms a characteristic infection peak. This corresponds to a population for which the frustration of social distancing is simply too great to allow the level of social behavior change required to mitigate the disease's spread.
\end{enumerate}

Note that the parameter bounds for these three cases can be stated in terms of the basic reproductive number $\mathcal{R}_0 = \frac{\beta}{\gamma}$. Defining $\mathcal{R}_{crit} = \frac{1 + \omega^*}{\omega^*}$, we have that the system \eqref{feedback-de} exhibits no outbreak if $0 < \mathcal{R}_0 < 1$, a controlled outbreak if $1 < \mathcal{R}_0 < \mathcal{R}_{crit}$, and an uncontrolled outbreak if $\mathcal{R}_{crit} < \mathcal{R}_0$. Intuitively, diseases with high reproductive numbers require less frustration with social distancing in order to revert to full outbreaks. In contrast, outbreaks of diseases with low basic reproductive numbers and high levels of frustration can still be controlled since the disease does not spread as fast. For diseases with high reproductive numbers, it is crucial that mitigation strategies are implemented in such a way that frustrations with disease mitigation strategies do not feedback to cause a decrease in social distancing.

\subsection{Secondary Waves of Infection}
\label{sec:waves}

\begin{figure}[t!]
\begin{subfigure}{.32\textwidth}
\centering
\includegraphics[width=2in]{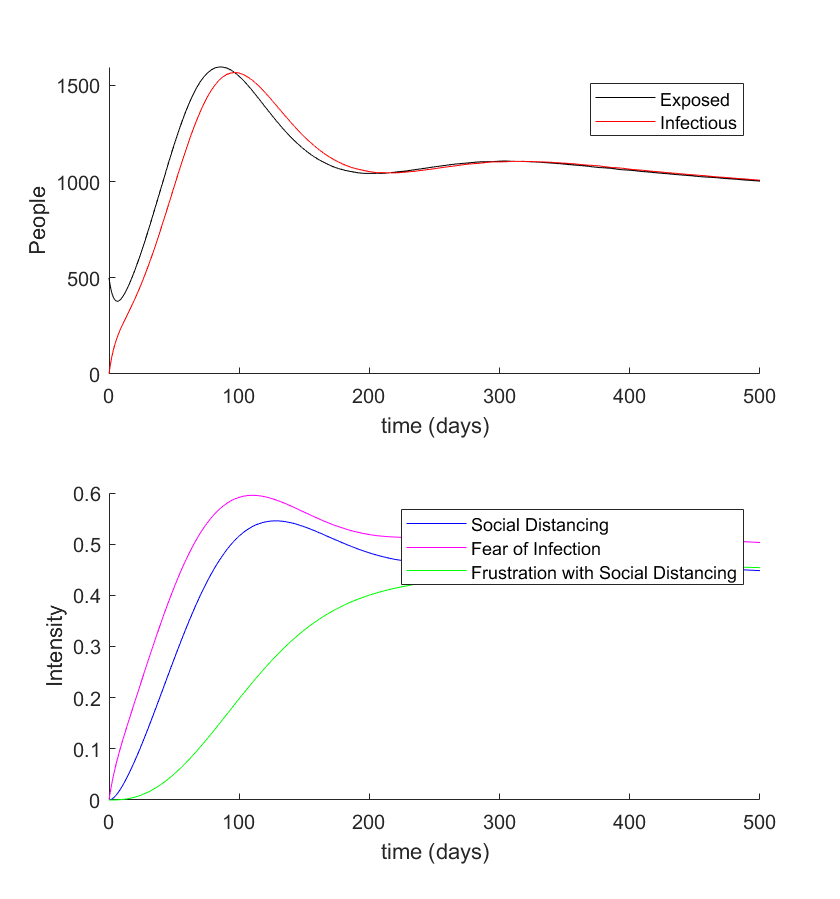}
\caption{$q = 1$ (No oscillations)}
\end{subfigure}
\begin{subfigure}{.32\textwidth}
\centering
\includegraphics[width=2in]{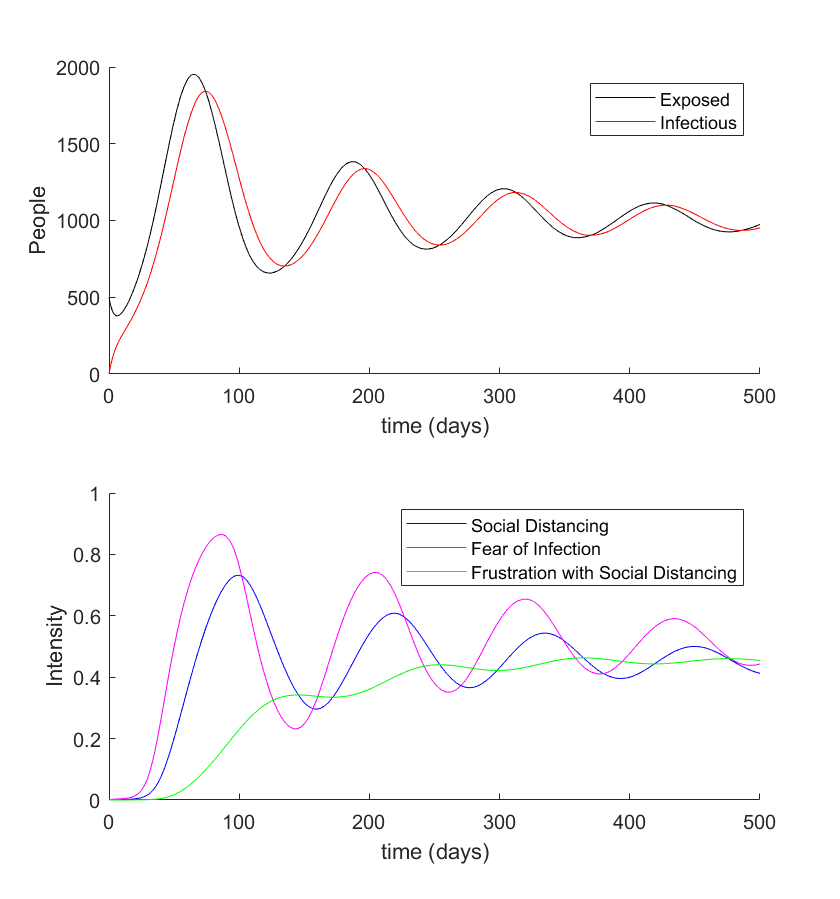}
\caption{$q = 5$ (transient oscillations)}
\end{subfigure}
\begin{subfigure}{.32\textwidth}
\centering
\includegraphics[width=2in]{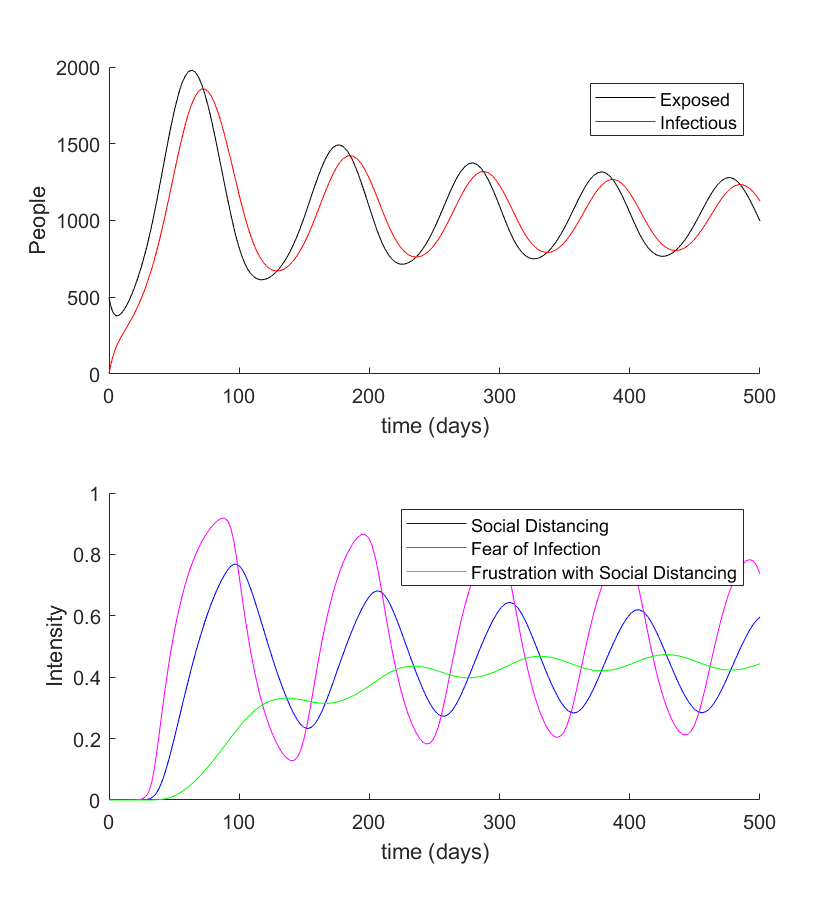}
\caption{$q = 10$ (sustained oscillations)}
\end{subfigure}
\caption{\footnotesize Simulations of the SEIR behavior-perception system \eqref{feedback-de} with the parameter values $\beta = 0.2$, $\lambda = 0.1$, $\gamma = 0.1$, $M = 100$, $k_{\omega} = 0.05$, $k_{P_I} = 0.05$, $k_{P_{\omega}} = 0.01$, $\omega^* = 0.25$, and $N = 1000000$ and initial condition $S(0) = 999500$, $E(0) = 500$, $I(0) = 0$, $R(0) = 0$, $\omega(0) = 0$, $P_I(0) = 0$, and $P_\omega(0) = 0$. Simulations are run for three cases on the parameter $q$: (a) $q = 1$; (b) $q = 5$; and (c) $q= 10$. We can see that, as $q$ becomes larger, the dynamics changes from quasi-stability (Figure (a)), to transient oscillations (Figure (b)), to sustained oscillations (Figure (c)). This suggests that the sharpness in the feedback from the infection level to social distancing behavior plays a key role in creating secondary waves of infection.}
\label{figure2}
\end{figure}

Significant discussion and public policy planning has centered around the possibility of secondary waves of infection. It is widely feared that after initial success in mitigating the spread of COVID-19, social perceptions and behaviors will change in response to falling infection numbers and that this will lead to a second wave of infection, and indeed this is suspected to be occurring in several regions around the world. Given the costs associated with social restrictions, predicting and assessing the timing and scale of potential second waves of infection is an area of significant research and concern.

To analyze the possibility of secondary waves of infection in the SEIR behavior-perception system \eqref{feedback-de}, we further analyze the stability of the endemic steady state \eqref{endemic-ss}. The mathematical analysis is contained in Appendix \ref{app:stability} and the results are summarized in Table \ref{table3}. We see that no oscillations are possible for the reduced no delay system \eqref{reduced-direct}, transient but not sustained oscillations (i.e. limit cycles) are possible in certain parameter regions for the one delay system \eqref{seir-onedelay-reduced}, and both transient and sustained oscillations are possible in certain parameter regions for the two-delay system \eqref{reduced-de}. This analysis suggests that delays in the social feedback and behavior variables play a key role in the generating the capacity for secondary waves of infection. The distinctions between no oscillations, transient oscillations, and sustained oscillations for the full SEIR behavior-perception model \eqref{feedback-de} are illustrated in Figure \ref{figure2}.  

\section{Results}
\label{sec:data}

In this section, we fit the SEIR behavior-perception system \eqref{feedback-de} to COVID-19 case incidence and mortality data. Data is taken from the database hosted by the Center for
Systems Science and Engineering
(CSSE) at Johns Hopkins University \cite{dong2020}. Model parameters and initial conditions for $\omega$, $P_I$ and $P_\omega$ were estimated through nonlinear least-squares curve fitting to the cumulative reported case and mortality data and their corresponding rates of change. We used the built-in MATLAB functions fmincon and multistart to minimize the following objective function:

$$Obj(\theta)=\sum_{i=1}^T \left(c(\theta,t_i)-\hat{c}_i)^2+(d(\theta,t_i)-\hat{d}_i)^2\right)+\sum_{i=T-10}^T \left(r_c(\theta,t_i)-\hat{r}_{ci})^2+(r_d(\theta,t_i)-\hat{r}_{di})^2\right),$$where $c(\theta,t_i)$ and $d(\theta,t_i)$ are the estimated cumulative cases and death cases from the model with parameter set $\theta$ at the calendar date $t_i$, respectively. $\hat{c}_i$ and $\hat{d}_i$ correspond to the actual reported cumulative cases and death cases at calendar date $t_i$. The second term consists of the rates of change of cumulative cases and death cases. That is, $c(\theta,t_i)$ and $d(\theta,t_i)$ are the rates of change of cumulative cases and death cases and $\hat{r}_{ci}$ and $\hat{r}_{di}$ are the actual rates of change from the data. $T$ is the total number of data points. We chose to include fitting to the number of deaths to promote constraining the model in the fitting process.

To test the model, we have selected six regions which have experienced different profiles in how the epidemic has spread: Canada, Italy, the United States, Israel, Michigan, and California. In Canada and Italy, after a substantial initial outbreak, the outbreak has been largely contained. In the United States, there was a large outbreak in March and April, a moderate dip in May and early June, and then a larger outbreak in July. In Israel, there was a small outbreak in March and April, a period of near eradication in May and June, and then a large outbreak in July. In Michigan, there was an initial outbreak that was largely contained, until early July where new cases of infection started to increase. California, much like the United States experienced a initial outbreak followed by a moderate dip in the month of April, and then entered the beginning of a much larger outbreak. The results are contained in Figures \ref{figure4} and \ref{figure5} and parameter values are contained in Table \ref{table4}. 

Simulated model fits indicate that the SEIR behavior-perception model supports different epidemic profiles as fear of infection, social distancing and frustration with social distancing are allowed to dynamically change between 0 and 1. These profiles include controlled disease incidence with near eradication (Canada in Figure~\ref{figure4}), loss of disease control and management (United States in Figure~\ref{figure4}), and disease resurgence in the form of a secondary wave (Israel in Figure~\ref{figure4}). 

\begin{figure}[t]
    \centering
    \begin{subfigure}{.32\textwidth}
\centering
\includegraphics[width=1.9in]{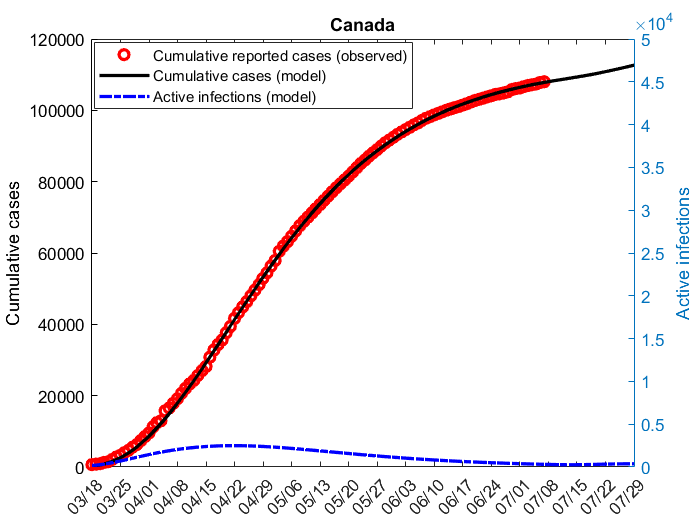}
\includegraphics[width=1.9in]{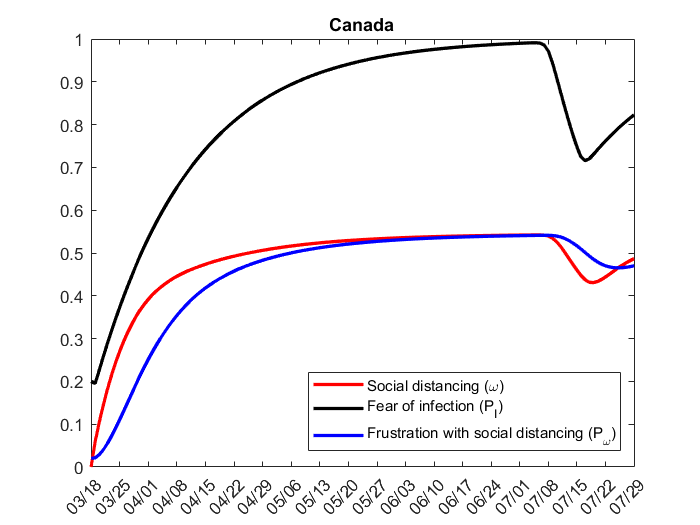}
\caption{Canada}
\end{subfigure}
    \begin{subfigure}{.32\textwidth}
\centering
\includegraphics[width=1.9in]{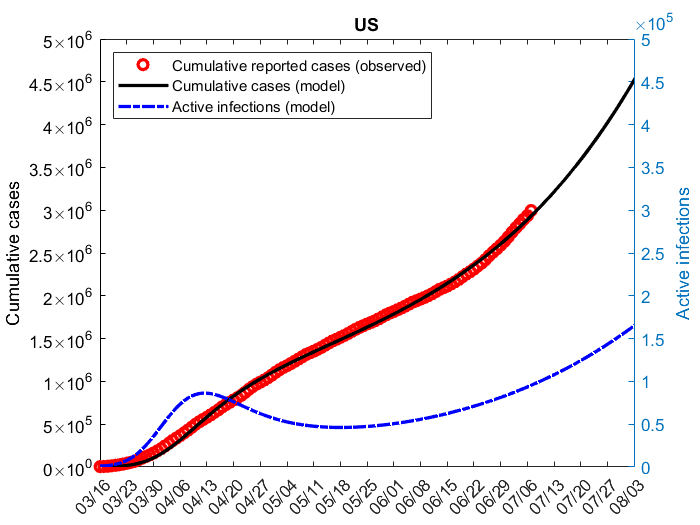}
\includegraphics[width=1.9in]{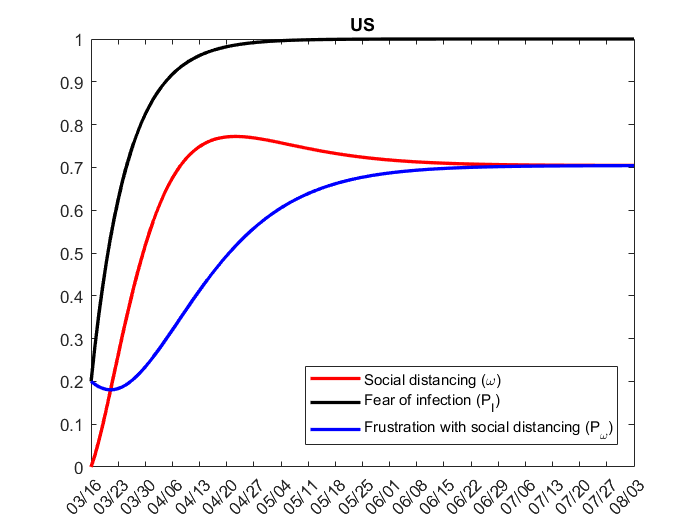}
\caption{United States}
\end{subfigure}
    \begin{subfigure}{.32\textwidth}
\centering
\includegraphics[width=1.9in]{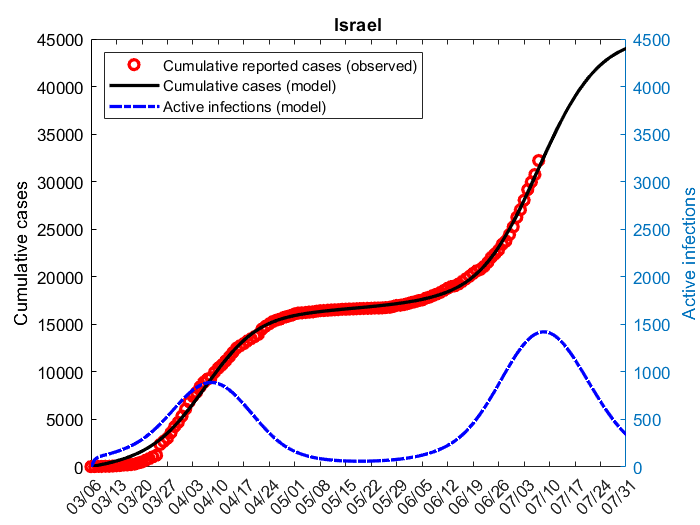}
\includegraphics[width=1.9in]{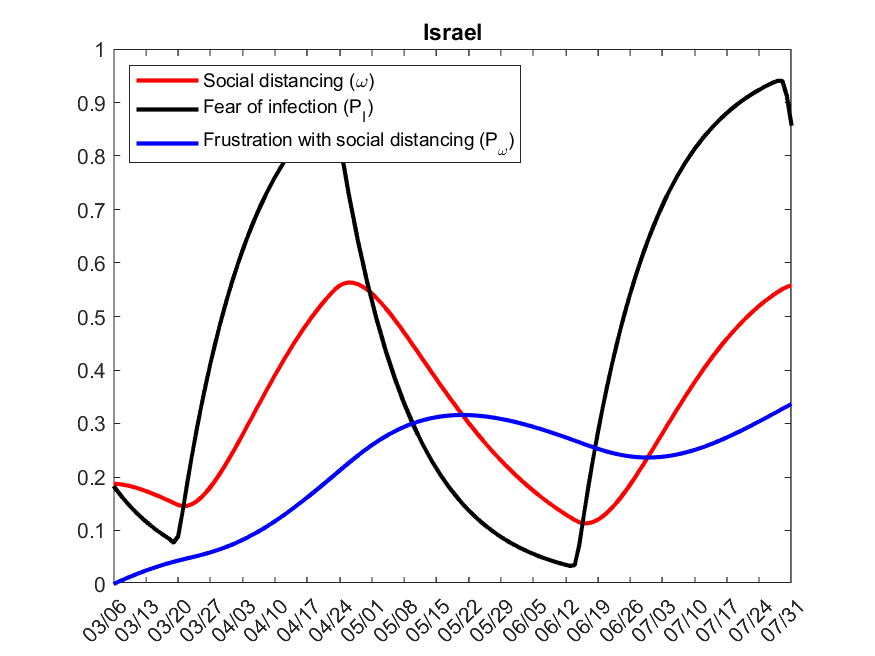}
\caption{Israel}
\end{subfigure}
    \caption{\footnotesize Model fits of the SEIR behavior-perception system model \eqref{feedback-de} to data sets of varying dynamics: (a) Canada, (b) United States and (c) Israel. Fitted Model parameters are shown in Table~\ref{table4}. Top row: cumulative cases and active cases (scaled according to the left and right axis, respectively). Bottom row: social distancing, perception of fear of infection and frustration with social distancing ($\omega$, $P_{I}$ and $P_{\omega}$). The basic reproductive numbers for Canada, US and Israel are 1.96, 3.78 and 1.53, respectively.   
    }
    \label{figure4}
\end{figure}


\begin{figure}[t!]
    \centering
    \begin{subfigure}{.32\textwidth}
\centering
\includegraphics[width=1.9in]{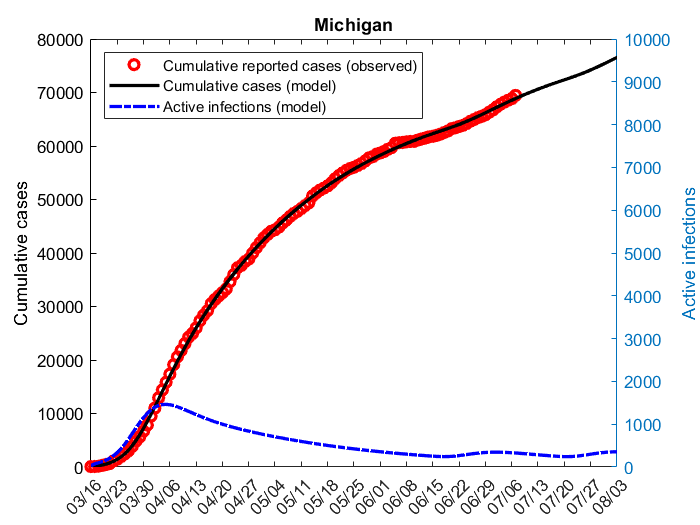}
\includegraphics[width=1.9in]{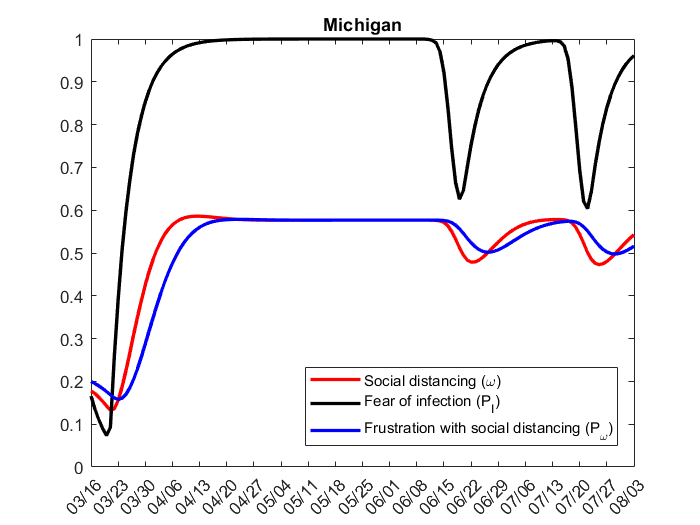}
\caption{Michigan}
\end{subfigure}
    \begin{subfigure}{.32\textwidth}
\centering
\includegraphics[width=1.9in]{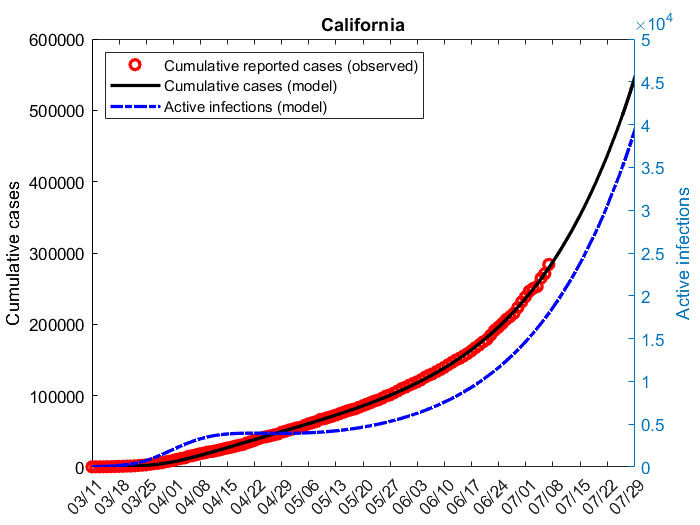}
\includegraphics[width=1.9in]{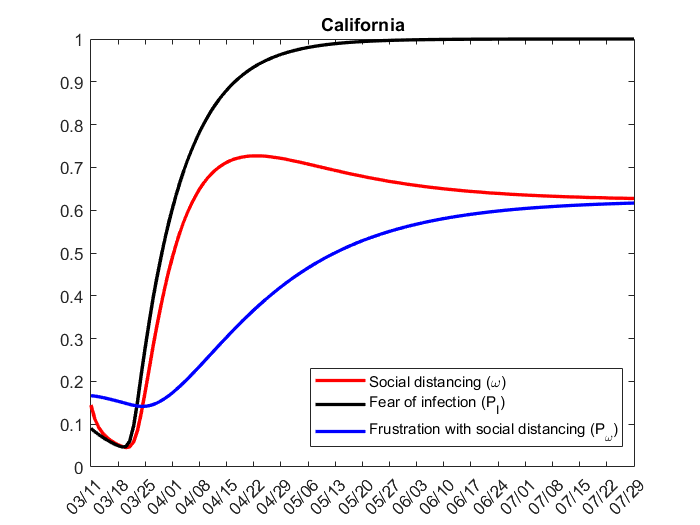}
\caption{California}
\end{subfigure}
    \begin{subfigure}{.32\textwidth}
\centering
\includegraphics[width=1.9in]{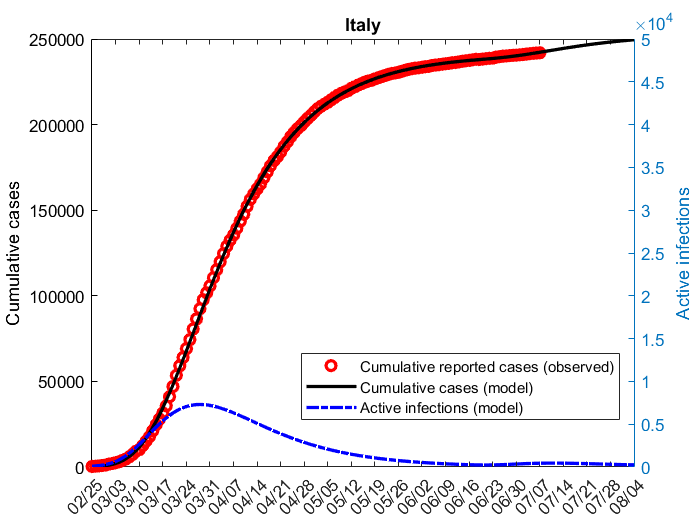}
\includegraphics[width=1.9in]{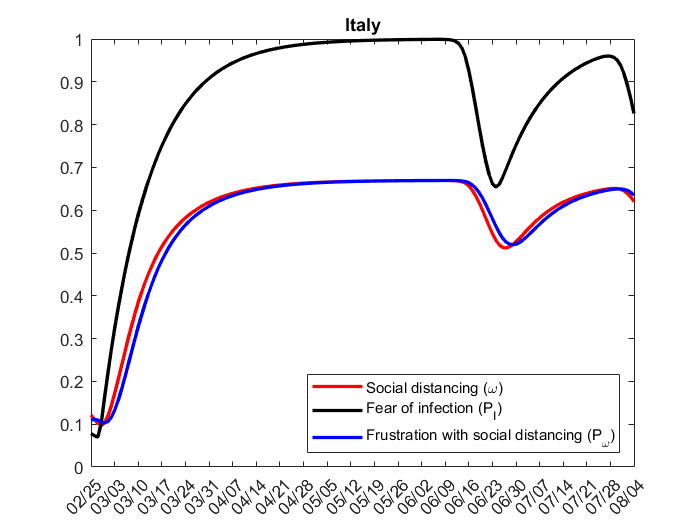}
\caption{Italy}
\end{subfigure}
    \caption{\footnotesize Model fits of the SEIR behavior-perception system model \eqref{feedback-de} to data sets of varying dynamics: (a) Michigan, (b) California and (c) Italy. Fitted Model parameters are shown in Table~\ref{table4}. Top row: cumulative cases and active cases (scaled according to the left and right axis, respectively). Bottom row: social distancing, perception of fear of infection and frustration with social distancing ($\omega$, $P_{I}$ and $P_{\omega}$). The basic reproductive numbers for Michigan, California and Italy are 2.2, 3.62 and 2.51, respectively. }
    \label{figure5}
\end{figure}

\begin{table}[t!]
\centering
\begin{tabular}{c|c|c|c|c|c|c}
\hline \hline
Parameter & Canada & United States & Israel & Michigan & California & Italy\\
\hline \hline
$\beta$     &   1.369    &    1.672     &   0.904 & 2.153& 1.470 & 1.788   \\
$\lambda$     &   0.701   &     0.370   & 0.6753 & 0.488 & 0.160& 0.456 \\
$\gamma$     & 0.697    &      0.442        & 0.589 & 0.978 & 0.405 & 0.711 \\
$k_{\omega}$     & 0.368 &   0.132     & 0.053 &00.136 & 0.651& 0.230 \\
$k_{P_{I}}$     &  0.042    & 0.108       & 0.064 & 0.202 & 0.085& 0.071 \\
$k_{P_{\omega}}$    &  0.127 &  0.037 & 0.021 & 0.220 & 0.023 & 0.499 \\
$q$     & 40.09 & 32.45 & 40.04 & 50 & 6.905& 17.02\\
$M$     & 198.0 & 73.97 & 195.9 & 240.2 & 233.1 & 180.7 \\ 
$\omega^*$     & 0.836 & 0.419 &0.954 & 0.733& 0.604 & 0.492 \\
\hline
$N$     & 37590000  & 328200000  & 9152000 & 99870000  &39510000 & 60360000 \\
$S(0)$     & 37589621 & 328198408 & 9151800 &  9986769 &39509766 & 60359707\\
$E(0)$     &200  & 200 & 200 &200 & 200 & 200 \\
$I(0)$     & 179 & 1392 & 0 & 31 & 34 & 93 \\
$R(0)$     & 0 & 0 & 0 & 0 & 0 & 0 \\
$\omega(0)$     & 0 & 0 & 0.187 & 0.177 &0.144 & 0.121 \\
$P_{I}(0)$     & 0.2 & 0.2 &0.183 & 0.165 & 0.089 &0.078\\
$P_{\omega}(0)$     & 0.019 & 0.2 &0 & 0.2 & 0.166& 0.109\\
\hline \hline
\end{tabular}
\caption{\footnotesize Parameter and initial values used to create model simulations in Figure~\ref{figure4} and \ref{figure5}.}
\label{table4}
\end{table}
Figure~\ref{figure4} shows model fits to reported case data for Canada, the United States and Israel. The model fit for Canadian reported case data shows an initial outbreak in early March which is ultimately controlled over the following four months. The model predicts stable levels of high fear of infection while social distancing and frustration with social distancing both tend to sufficient levels for disease control and near eradication. On the other hand, the United States shows an initial outbreak which is temporarily controlled until late May. Fear of infection quickly saturates to 1, while frustration with social distancing slowly increases over time. After the initial increase in social distancing, the increasing levels of frustration lead to a steady decrease in social distancing levels which become insufficient for control or eradication and thus active cases of infection increase. Interestingly, Israel shows an initial outbreak followed by a period of near eradication of the disease. However, fear of infection and social distancing decrease to insufficient levels too quickly after the first wave of infections and leads to a secondary outbreak.

Figure~\ref{figure5} shows model fits of the SEIR behavior-perception system for Michigan, California and Italy. Michigan shows an initial outbreak that is relatively under control. However, the model predicts a period starting after mid-July where social distancing, fear of infection and frustration with social distancing begin to oscillate, but not to the extend seen in the Israel model fit from Figure~\ref{figure4}. California experiences an initial outbreak that is only temporarily managed before increasing exponentially. After the initial outbreak, the rise in active cases is accompanied by high fear of infection and an increasing level of frustration with social distancing. This leads to a reduction in social distancing. Italy dynamically behaves similar to what is seen in Michigan and Canada. That is, after the initial outbreak, high fear of infection drives a period of stable social distancing and frustration with social distancing.
\section{Discussion and Conclusions}
\label{sec:conclusions}

We have introduced an SEIR behavior-perception system \eqref{feedback-de} for modeling the feedback of social fear of infection and frustration with social distancing on the dynamics of disease spread. We have shown the following: (1) including fear of infection leads to a \emph{controlled outbreak} the level of which depends on each society's tolerance for infection; (2) delays in the fear of infection can lead to \emph{secondary and sustained waves of infection}; and (3) frustration with social distancing can overcome the fear of infection to lead to an \emph{uncontrolled outbreak}. Where analytically possible, we have provided parameter ranges where the relevant behaviors occur.

We have also fit the model to cumulative COVID-19 case data from several regions: Canada, the United States, Israel, Michigan, California and Italy. (See Figures \ref{figure4} and \ref{figure5}.) The fits obtained validate our SEIR behavior-perception model as capable of capturing the emergence of social-feedback-driven secondary waves of infection. Our analysis furthermore suggests that regions which experience significant reductions in new infection levels following their initial outbreak, such as Israel and Canada, are likely to be able to mitigate secondary waves of infection. Regions which experiences only moderate reductions in new infection levels, such as the United States, are likely to experience more dramatic secondary waves and are at increase risk for entering into an uncontrolled outbreak.

Using the full model~\eqref{feedback-de}, we estimated the basic reproductive number for Canada, United States, Israel, Michigan, California and Italy, to be 1.96, 3.78, 1.53, 2.2, 3.62 and 2.51 respectively. These estimations are inline with other studies \cite{Zhuang2020,Lv2020}. In addition, the time varying effective reproductive number is now implicitly a function of the social behavior in the population, because of the incorporation of behavior-perception feedback variables, $\omega$, $P_I$ and $P_{\omega}$. This effective reproductive number may provide further real-time and future insight into not only the disease dynamics, but the social dynamics that drive the spread of the disease. 

The reduced model \eqref{reduced-direct} produces a threshold number, $\mathcal{R}_{crit}>1$ such that when $\mathcal{R}_0>\mathcal{R}_{crit}$ the solutions become unbounded. This unbounded behavior manifests itself within the full model \eqref{feedback-de} as the uncontrolled outbreak which forms a characteristic peak often seen in disease outbreaks. However, unlike in the reduced model, the full model has bounded solutions, because of the finite susceptible population. Thus herd immunity ultimately reduces the disease burden to zero and the disease dies out. These dynamics are illustrated in Figure~\ref{figure1}.

The model introduced in this paper presents several immediate opportunities for further work.
\begin{enumerate}
    \item 
    The model \eqref{feedback-de} assumes the same critical infection level $M$ for both the ascending and descending phases of disease dynamics. In practice, however, we have seen some countries quick to lockdown and slow to open up, while other countries are slow to lockdown and quick to open up, which suggests a different critical value $M$ in the ascending and descending phase of an outbreak.
    \item 
    Currently, our estimates for the social distancing variable $\omega$ only come through disease prevalence. Throughout the COVID-19 pandemic, however, data on social mobility has been provided by several sources, including Google, Apple, and the United States Department of Transportation~\cite{GoogleMobility,AppleMobility}. Further insight in the social mechanisms underlying disease dynamics might be gained by fitting $\omega$ to this data as well.
    \item
    Numerical simulations such as those in Figure \ref{figure2} suggest that the reduced SEIR behavior-perception model \eqref{reduced-de} undergoes a Hopf bifurcation as the endemic steady state \eqref{endemic-ss} loses stability. Numerical results also suggest that trajectories of the reduced models \eqref{seir-onedelay-reduced} and \eqref{reduced-de} are bounded in the controlled outbreak scenario but unbounded for the uncontrolled outbreak scenario. These results are currently unproved.
    %
\end{enumerate}



\bibliographystyle{plain}

\begin{appendices}

\section{Stability analysis of SEIR behavior-perception system}
\label{app:stability}

In this Appendix, we provide the mathematical stability analysis for the disease-free and endemic steady states of the reduced feedback systems \eqref{reduced-direct}, \eqref{seir-onedelay-reduced}, and \eqref{reduced-de}. We make use of the trace-determinant condition for planar systems \cite{Strogatz} and the Routh-Hurwitz criterion for non-planar systems \cite{Routh,Hurwitz}. The Routh-Hurwitz conditions states that a given polynomial has all roots with negative real part if and only if the first column entries of the Routh-Hurwitz table are all positive. Consequently, if the Routh-Hurwitz table of the characteristic polynomial of a matrix has any first column entry which is negative, then the matrix has an eigenvalue with positive real part.

\subsection{No delays}
\label{app:nodelay}

Assuming that $\omega$, $P_I$, and $P_{\omega}$ (if considered) equilibriate immediately, \eqref{reduced-de} reduces to the following direct systems:
\begin{equation} \small
    \label{seir-direct-reduced}
    \begin{array}{cc}  \mbox{\textbf{\underline{No frustration}}} & \mbox{\textbf{\underline{With frustration}}} \\ \\
    \left\{ \; \;
    \begin{split}
\frac{dE}{dt} & = \beta\left(1 - \frac{(\lambda E)^q}{M^q + (\lambda E)^q}\right) I - \lambda E \\
    \frac{dI}{dt} & = \lambda E - \gamma I.
    \end{split} \right.
 &
 \left\{ \; \;
    \begin{split} 
    \frac{dE}{dt} & = \beta\left(1 - \frac{(\lambda E)^q}{M^q + (1 + \omega^*)(\lambda E)^q}\right) I - \lambda E \\
    \frac{dI}{dt} & = \lambda E - \gamma I.
    \end{split} \right.
    \end{array}
\end{equation}
In addition to the disease-free steady state $(\bar{E},\bar{I})_{df} = (0,0)$, the system \eqref{seir-direct-reduced} has the endemic steady state
\begin{equation}
\label{endemic-ss1}
    (\bar{E},\bar{I})_{end}  =  \left( \frac{M}{\lambda} \left( \frac{\beta-\gamma}{\gamma - \omega^*(\beta - \gamma)} \right)^{\frac{1}{q}}, \frac{M}{\gamma} \left( \frac{\beta-\gamma}{\gamma - \omega^*(\beta - \gamma)} \right)^{\frac{1}{q}} \right)
\end{equation}
which is only physical meaningful if $0 < \beta - \gamma < \frac{\gamma}{\omega^*}$. We will consider the ``With frustration'' system in \eqref{seir-direct-reduced} and reduced to the ``No frustration'' case by taking $\omega^*=0$. 
 

We have the following result.



\begin{theorem}
\label{theorem1}
Consider the reduced direct SEIR behavior-perception models \eqref{seir-direct-reduced}. The following behaviors are possible:
\begin{enumerate}
    \item 
    If $\beta < \gamma$ then there is only the disease-free steady state and it is asymptotically stable.
    \item
    If $0 < \beta - \gamma < \frac{\gamma}{\omega^*}$ then the disease-free steady state is unstable, and the endemic steady state is positive and asymptotically stable. Furthermore, trajectories near the endemic steady state may not exhibit oscillatory behavior.
    \item
    If $\beta - \gamma > \frac{\gamma}{\omega^*}$ then there is only the disease-free steady state and it is unstable, and solutions become unbounded.
\end{enumerate}
\end{theorem}

\begin{proof}

We consider the Jacobian of the system \eqref{seir-direct-reduced} evaluated at the disease-free steady state and the endemic steady state \eqref{endemic-ss1}. After simplifying, we have
\begin{equation}
    \label{jacobian-seir-direct-reduced}
    J_{df} = \left[ \begin{array}{cc} -\lambda & \beta \\ \lambda & - \gamma \end{array} \right] \; \mbox{ and } \; J_{end} = \left[ \begin{array}{cc} -\lambda \left( \frac{(\beta-\gamma)(\gamma-\omega^*(\beta-\gamma))q- \gamma \beta}{\gamma \beta}\right) & \gamma \\ \lambda & - \gamma \end{array} \right].
\end{equation}
We have $\mbox{tr}(J_{df}) = -(\lambda + \gamma) <0$ and $\mbox{det}(J_{df}) = -\lambda(\beta-\gamma)$. It follows that disease-free steady state is a stable node if $\beta < \gamma$ (no outbreak) and a saddle if $\beta > \gamma$ (controlled or uncontrolled outbreak). For the endemic steady state, we have
\[
\begin{split}
    \mbox{det}(J_{end}) & = \frac{(\beta - \gamma)(\gamma - \omega^*(\beta - \gamma))\lambda q}{\beta}\\
    \mbox{tr}(J_{end}) & = -\frac{(\beta - \gamma)(\gamma - \omega^*(\beta - \gamma)) \lambda q + \gamma \beta (\gamma + \lambda)}{\gamma \beta}.
\end{split}
\]
Since we are only concerned with the endemic steady state when $0 < \beta - \gamma < \frac{\gamma}{\omega^*}$, we have that det$(J_{end}) > 0$ and tr$(J_{end}) < 0$. It follows that the endemic steady state \eqref{endemic-ss1} of \eqref{seir-direct-reduced} is asymptotically stable whenever it exists.

In order for trajectories near the endemic steady state \eqref{endemic-ss1} to oscillate, we require that tr$(J_{end})^2 - 4$det$(J_{end}) < 0$. To analyze this condition, we compute
\begin{equation}
    \label{tr-det}
\mbox{tr}(J_{end})^2 - 4\mbox{det}(J_{end}) = \frac{[\lambda (\beta-\gamma)(\gamma - \omega^*(\beta-\gamma))q + \gamma \beta (\lambda - \gamma)]^2 + 4 \lambda \gamma^3\beta^2}{\gamma^2 \beta^2} > 0.
\end{equation}
It follows that the linearized system may not permit oscillations around the endemic steady state \eqref{endemic-ss1} so that \eqref{seir-direct-reduced} does not permit the oscillatory behavior.

To prove solutions become unbounded if $\beta - \gamma > \frac{\gamma}{\omega^*}$ (uncontrolled outbreak) we note that this condition implies that
\[(\lambda E)^q > 0 > \frac{(\beta - \gamma)M^q}{\gamma - \omega^*(\beta-\gamma)}\]
for $E > 0$. It follows that
\[\frac{(\lambda E)^q}{M^q + (1 + \omega^*) (\lambda E)^q} < \frac{\beta - \gamma}{\beta}\]
so that
\[\beta \left( 1 - \frac{(\lambda E)^q}{M^q + (1 + \omega^*) (\lambda E)^q} \right) > \gamma.\]
Now consider the quantity $E + I$. From \eqref{reduced-direct} we have that
\[\frac{d}{dt} ( E + I) = \left[ \beta \left( 1 - \frac{(\lambda E)^q}{M^q + (1 + \omega^*) (\lambda E)^q} \right) -  \gamma \right] I > 0.\]
It follows that, provided $E > 0$, then we have that $E + I$ is continually increasing for all time. Now suppose that there is a least upper bound $E_{lim} + I_{lim} > 0$ such that $E(t) + I(t) \leq E_{lim} + I_{lim}$. It follows that $E(t) + I(t)$ approaches a single point on this set and by continuity of \eqref{reduced-direct}, this point must be a steady state. When $\beta - \gamma > \frac{\gamma}{\omega^*}$, however, there is only the disease-free steady state and this does not lie on this set. It follows that there is not a least upper bound $E_{lim} + I_{lim}$ to $E(t) + I(t)$. Consequently, it follows that
\[\lim_{t \to \infty} E(t) + I(t) = \infty.\]
That is, the overall level of infection in the population is unbounded.
\end{proof}

\subsection{One Delay}
\label{app:onedelay}

Assuming that fear of infection ($P_I$) operates on a significantly faster timescale than the remainder of the variables, we can set $\frac{dP_I}{dt} = 0$ in \eqref{reduced-de} to get the \emph{reduced one delay SEIR behavior-perception systems}:
\begin{equation} \small
    \label{seir-onedelay}
    \begin{array}{cc}  \mbox{\textbf{\underline{No frustration}}} & \mbox{\textbf{\underline{With frustration}}} \\ \\
    \left\{ \; \;
    \begin{split}
    \frac{dE}{dt} & = \beta\left(1 - \omega \right) I - \lambda E \\
    \frac{dI}{dt} & = \lambda E - \gamma I\\
    \frac{d\omega}{dt} & = k_{\omega} \left( \frac{(\lambda E)^q}{M^q + (\lambda E)^q} - \omega \right)
    \end{split} \right.
 &
 \left\{ \; \;
    \begin{split} 
    \frac{dE}{dt} & = \beta\left(1 - \omega \right) I - \lambda E \\
    \frac{dI}{dt} & = \lambda E - \gamma I\\
    \frac{d\omega}{dt} & = k_{\omega} \left( \frac{(\lambda E)^q}{M^q + (\lambda E)^q}(1 - \omega^* P_\omega) - \omega \right) \\
    \frac{dP_{\omega}}{dt} & = k_{P_{\omega}} \left( \omega - P_{\omega} \right).
    \end{split} \right.
    \end{array}
\end{equation}
The system \eqref{seir-onedelay} has the disease-free steady state $(\bar{E},\bar{I},\bar{\omega},\bar{P_{\omega}})_{df} = (0,0,0,0)$ and the endemic steady state
\begin{equation}
    \label{endemic-ss2}
    (\bar{E},\bar{I},\bar{\omega},\bar{P}_{\omega})_{end} = \left( \frac{M}{\lambda} \left( \frac{\beta-\gamma}{\gamma- \omega^*(\beta-\gamma)} \right)^{\frac{1}{q}}, \frac{M}{\gamma} \left( \frac{\beta-\gamma}{\gamma- \omega^*(\beta-\gamma)} \right)^{\frac{1}{q}}, \frac{\beta - \gamma}{\beta}, \frac{\beta - \gamma}{\beta} \right)
\end{equation}
which again is only physical meaningful if $0 < \beta - \gamma < \frac{\gamma}{\omega^*}$. Once again, we will consider only the ``with frustration'' case and limit to the ``no frustration'' case by taking $\omega^* = 0$. 

We have the following result.

\begin{theorem}
\label{theorem2}
Consider the reduced direct SEIR behavior-perception models \eqref{seir-direct-reduced}. The following behaviors are possible:
\begin{enumerate}
    \item 
    If $\beta < \gamma$ then there is only the disease-free steady state and it is asymptotically stable.
    \item
    If $0 < \beta - \gamma < \frac{\gamma}{\omega^*}$ then the disease-free steady state is unstable, and the endemic steady state is positive and asymptotically stable. Furthermore, trajectories near the endemic steady state may or may not exhibit oscillatory behavior depending on the parameter values.
    \item
    If $\beta - \gamma > \frac{\gamma}{\omega^*}$ then there is only the disease-free steady state and it is unstable.
\end{enumerate}
\end{theorem}

\begin{proof}
The Jacobian of \eqref{seir-onedelay-reduced} evaluated at the disease-free steady state is
\begin{equation}
    \label{jacobian-df}
J_{df} = \left[\begin{array}{cccc} -\lambda & \beta & 0 & 0 \\ \lambda & -\gamma & 0 & 0 \\ 0 & 0 & -k_{\omega} & 0 \\ 0 & 0 & k_{P_{\omega}} & -k_{P_{\omega}}
\end{array} \right].
\end{equation}
The diagonal structure of \eqref{jacobian-df} means the eigenvalues can be analysed by considering the decomposed matrices
\begin{equation}
\label{decomposed1}
\left[ \begin{array}{cc} - \lambda & \beta \\ \lambda & - \gamma \end{array} \right] \; \mbox{ and } \; \left[ \begin{array}{cc} -k_{\omega} & 0 \\ k_{P_{\omega}} & -k_{P_{\omega}} \end{array} \right].
\end{equation}
The right matrix in \eqref{decomposed1} has the eigenvalues $\lambda_{1,2} = -k_{\omega}, -k_{P_{\omega}} < 0$ while the eigenvalues of the left matrix in \eqref{decomposed1} were considered in the proof of Theorem \ref{theorem1}. It follows that the disease-free steady state is a stable node if $\beta < \gamma$ and a saddle if $\beta > \gamma$.

The Jacobian of \eqref{seir-onedelay-reduced} evaluated at the endemic steady state $(\bar{E},\bar{I},\bar{\omega},\bar{P_I})$ \eqref{endemic-ss2} is
\begin{equation}
    \label{jacobian2}
J_{end} = \left[\begin{array}{cccc} -\lambda & \gamma & -\frac{\beta M}{\gamma} \left( \frac{\beta - \gamma}{\gamma - \omega^*(\beta - \gamma)} \right)^{\frac{1}{q}} & 0 \\ \lambda & -\gamma & 0 & 0 \\ \frac{\lambda k_{\omega}(\beta - \gamma)^{1 - \frac{1}{q}} (\gamma - \omega^*(\beta - \gamma))^{1 + \frac{1}{q}}}{\beta M ( \beta - \omega^*(\beta - \gamma))} & 0 & -k_{\omega} & -\frac{\omega^*(\beta-\gamma)k_{\omega}}{\beta - \omega^*(\beta-\gamma)} \\
0 & 0 & k_{P_{\omega}} & - k_{P_{\omega}}
\end{array} \right]
\end{equation}
It is not feasible to determine the eigenvalues of \eqref{jacobian2} directly. The Routh-Hurwitz criterion, however, can be applied to the characteristic polynomial of \eqref{jacobian2}. The first two entries of the Routh-Hurwitz table are $1$ and $\lambda + \gamma + k_{P_{\omega}}+k_{\omega}$, which are trivially positive. The third and fourth entries can be expanded and factored into terms which are either positive of contain factors of $\beta - \gamma$ of $\beta - \omega^*(\beta - \gamma)$ (see supplemental computational material). We notice that the endemic condition $0 < \beta - \gamma < \frac{\gamma}{\omega^*}$ implies that
\[\beta - \omega^*(\beta - \gamma) > \gamma - \omega^*(\beta - \gamma) > \gamma - \omega^* \left( \frac{\gamma}{\omega^*}\right) > 0\]
so that the third and fourth Routh-Hurwitz entries are positive. The fifth entry is
\[\frac{k_{P_{\omega}} k_{\omega}\lambda(\beta-\gamma)(\gamma-\omega^*(\beta-\gamma) )q}{\beta-\omega^*(\beta-\gamma)}\]
This is positive, so that all the first-column entries of the Routh-Hurwitz table are positive. It follows that all of the eigenvalues of \eqref{jacobian-seir-direct-reduced} have negative real part so that the endemic steady state of \eqref{seir-onedelay} is locally asymptotically stable. It follows that the system with one delay is not consistent with sustained oscillations around the endemic steady state \eqref{endemic-ss2}.

It is challenging to find explicit conditions on the parameters which guarantee \eqref{jacobian2} has complex eigenvalues. It can be checked numerically, however, that this is possible. One example set of parameters is $\beta = 0.2$, $\lambda = 0.1$, $\gamma = 0.1$, $q = 10$, $M = 100$, $k_{\omega} = 0.1$, $k_{P_{\omega}} = 0.1$, $\omega^* = 0$. Substituted in \eqref{jacobian2}, this produces the eigenvalues: $\lambda_1 = -0.1$, $\lambda_2 = -0.1$, $\lambda_{3,4} = -0.1 \pm 0.2 i$. It follows that the one-delay models \eqref{seir-onedelay} permit transient waves of infection.
\end{proof}

\subsection{Two Delays}
\label{app:twodelays}

Now consider the full two delay model, with and without frustration with social distancing:
\begin{equation} \small
    \label{seir-full}
    \begin{array}{cc}  \mbox{\textbf{\underline{No frustration}}} & \mbox{\textbf{\underline{With frustration}}} \\ \\
    \left\{ \; \;
    \begin{split}
    \frac{dE}{dt} & = \beta\left(1 - \omega \right) I - \lambda E \\
    \frac{dI}{dt} & = \lambda E - \gamma I\\
    \frac{d\omega}{dt} & = k_{\omega} \left( P_I - \omega \right) \\
    \frac{dP_{I}}{dt} & = k_{P_I} \left( \frac{(\lambda E)^q}{M^q + (\lambda E)^q} - P_I \right)
    \end{split} \right.
 &
 \left\{ \; \;
    \begin{split} 
    \frac{dE}{dt} & = \beta\left(1 - \omega \right) I - \lambda E \\
    \frac{dI}{dt} & = \lambda E - \gamma I\\
    \frac{d\omega}{dt} & = k_{\omega} \left( P_I(1 - \omega^* P_\omega) - \omega \right)  \\
    \frac{dP_{I}}{dt} & = k_{P_I} \left( \frac{(\lambda E)^q}{M^q + (\lambda E)^q} - P_I \right)\\
    \frac{dP_{\omega}}{dt} & = k_{P_{\omega}} \left( \omega - P_{\omega} \right).
    \end{split} \right.
    \end{array}
\end{equation}
The system \eqref{seir-full} has the disease-free steady state $(\bar{E},\bar{I},\bar{\omega},\bar{P_I},\bar{P_{\omega}})_{df} = (0,0,0,0,0)$. The endemic steady state of \eqref{seir-full} is
\begin{equation}
    \label{endemic-ss3}
    \small
    (\bar{E},\bar{I},\bar{\omega},\bar{P_I},\bar{P_{\omega}})_{end} = \left( \frac{M}{\lambda} \left( \frac{\beta-\gamma}{\gamma-\omega^*(\beta - \gamma)} \right)^{\frac{1}{q}}, \frac{M}{\gamma} \left( \frac{\beta-\gamma}{\gamma-\omega^*(\beta - \gamma)} \right)^{\frac{1}{q}}, \frac{\beta - \gamma}{\beta},\frac{\beta-\gamma}{\beta - \omega^*(\beta - \gamma)}, \frac{\beta - \gamma}{\beta} \right)
\end{equation}
which is physically meaningful if and only if $0 < \beta - \gamma < \frac{\gamma}{\omega^*}$. We consider the ``with frustration'' case with the understanding that we can limit to the ``no frustration'' case by taking $\omega^* = 0$. 

We have the following result.

\begin{theorem}
\label{theorem2}
Consider the reduced direct SEIR behavior-perception models \eqref{seir-direct-reduced}. The following behaviors are possible:
\begin{enumerate}
    \item 
    If $\beta < \gamma$ then there is only the disease-free steady state and it is asymptotically stable.
    \item
    If $0 < \beta - \gamma < \frac{\gamma}{\omega^*}$ then the disease-free steady state is unstable, and the endemic steady state is positive. Trajectories near the endemic steady state may converge exponentially toward it, exhibit damped oscillations, or converge toward nearby limit cycles, depending on the parameter values.
    \item
    If $\beta - \gamma > \frac{\gamma}{\omega^*}$ then there is only the disease-free steady state and it is unstable.
\end{enumerate}
\end{theorem}

\begin{proof}
The Jacobian of \eqref{seir-full} evaluated at the disease-free steady state is
\begin{equation}
    \label{jacobian-df3}
    \left[
    \begin{array}{ccccc}
    -\lambda & \beta & 0 & 0 & 0 \\
    \lambda & - \gamma & 0 & 0 & 0 \\
    0 & 0 & -k_{\omega} & k_{\omega} & 0 \\
    0 & 0 & 0 & -k_{P_I} & 0 \\
    0 & 0 & k_{P_{\omega}} & 0 & -k_{P_{\omega}}
    \end{array} \right]
\end{equation}
The diagonal structure of \eqref{jacobian-df3} means the eigenvalues can be analysed by considering the decomposed matrices
\begin{equation}
\label{decomposed2}
\left[ \begin{array}{cc} - \lambda & \beta \\ \lambda & - \gamma \end{array} \right] \; \mbox{ and } \; \left[ \begin{array}{ccc} -k_{\omega} & k_{\omega} & 0 \\
0 & -k_{P_I} & 0 \\
k_{P_{\omega}} & 0 & -k_{P_{\omega}} \end{array} \right].
\end{equation}
The right matrix in \eqref{decomposed2} has the eigenvalues $\lambda_{1,2,3} = -k_{\omega}, -k_{P_I}, -k_{P_{\omega}} < 0$ while the eigenvalues of the left matrix in \eqref{decomposed2} were considered in the proof of Theorem \ref{theorem1}. It follows that the disease-free steady state is a stable node if $\beta < \gamma$ and a saddle if $\beta > \gamma$.

The Jacobian evaluated at the endemic steady state  $(\bar{E},\bar{I},\bar{\omega},\bar{P_I},\bar{P_{\omega}})_{end}$ \eqref{endemic-ss3} is
\begin{equation}
    \label{jacobian3}
    \left[
    \begin{array}{ccccc}
    -\lambda & \gamma & - \frac{M\beta}{\gamma} \left( \frac{\beta - \gamma}{\gamma-\omega^*(\beta-\gamma)} \right)^{\frac{1}{q}} & 0 & 0\\
    \lambda & -\gamma & 0 & 0 & 0 \\
    0 & 0 & -k_{\omega} & k_{\omega} \left(\frac{\beta - \omega^*(\beta-\gamma)}{\beta} \right) & -\frac{k_{\omega} \omega^* (\beta-\gamma)}{\beta - \omega^*(\beta-\gamma)} \\
    \frac{\lambda k_{P_I} q (\beta - \gamma)^{1 - \frac{1}{q}}(\gamma - \omega^*(\beta - \gamma)^{1 + \frac{1}{q}}}{M(\beta - \omega^*(\beta-\gamma))^2} & 0 & 0 & -k_{P_I} & 0 \\
    0 & 0 & k_{P_{\omega}} & 0 & - k_{P_{\omega}} \end{array} \right]
\end{equation}
The eigenvalues of \eqref{jacobian3} cannot be reasonably computed directly. The Routh-Hurwitz condition, however, can be utilized to determine if the matrix permits eigenvalues with positive real part. The first two column entries of the Routh-Hurwitz table are $1$ and $\lambda + \gamma +k_{P_I}+k_{P_{\omega}}+k_{\omega}$, both of which are positive. The third entry can be factored as
\begin{equation}
    \label{W2}\frac{-n(q,1)\cdot q + n(q,0)}{d(q,0)}
    \end{equation}
where
\[
\begin{split}
n(q,1) & =k_{P_I} \lambda k_{\omega}(\beta-\gamma)(\gamma-\omega^*(\beta-\gamma)))\\
n(q,0) & = \gamma(k_{P_I}+k_{\omega}+k_{P_{\omega}}(\lambda+\gamma+k_{P_{\omega}}+k_{\omega})(\lambda+k_{P_I}+\gamma)(\beta-\omega^*(\beta-\gamma))+\beta\gamma k_{P_{\omega}}k_{\omega}(k_{P_{\omega}} +k_{\omega})\\
d(q,0) & = \gamma(\gamma+k_{P_I}+k_{P_{\omega}}+k_{\omega}+\lambda)(\beta-\omega^*(\beta-\gamma))
\end{split}
\]
For the endemic condition $0 < \beta - \gamma < \frac{\gamma}{\omega^*}$ we have
\[\beta - \gamma > 0, \; \gamma - \omega^*(\beta - \gamma) > 0, \; \mbox{ and } \; \beta - \omega^*(\beta - \gamma) > 0\]
so that $n(q,1) > 0$, $n(q,0) > 0$, and $d(q,0) > 0$. It follows that \eqref{W2} can be made negative by choosing
\[q > \frac{n(q,0)}{n(q,1)}.\]
In particular, for any values of the parameters $\beta, \lambda, \gamma, k_{P_I},$ and $k_{\omega}$ we can choose $q$ sufficiently large so that the first column of the Routh-Hurwitz table has a sign change, so that the endemic steady state $(\bar{E},\bar{I},\bar{\omega},\bar{P_I},\bar{P_{\omega}})_{end}$ \eqref{endemic-ss3} is unstable. The value of $q$ increasing corresponds to the social behavior operating more and more like a switch. Independently, there is capacity for a sign change in the fourth Routh-Hurwitz coefficient; however, the algebra is too complicated to present here directly (see supplemental computational files for numerical results).
\end{proof}

\end{appendices}

\end{document}